\documentclass[a4paper,11pt]{article}
\usepackage[a4paper]{geometry}
\usepackage{fullpage}
\title{}
\date{} 
\author{}
\usepackage[style=numeric-comp,useprefix,hyperref,backend=bibtex]{biblatex}
\usepackage{graphicx}
\graphicspath{ {./images/} }
\usepackage{subfig}
\usepackage{framed}
\usepackage{amsmath}
\usepackage{amssymb}
\usepackage{amsfonts}
\usepackage{mathrsfs}
\usepackage{array}
\usepackage{enumerate}
\usepackage{amsthm}  
\usepackage{tikz}
\usepackage[linktocpage=true]{hyperref}
\usepackage{listings}
\usepackage{xcolor}
\usepackage{dsfont}
\usepackage{wrapfig}
\usepackage{float}
\usepackage{verbatim}

\usepackage{biblatex}
\theoremstyle{plain}
\newtheorem{thm}{Theorem}[section] 
\newtheorem{lem}[thm]{Lemma} 
\newtheorem{prop}[thm]{Proposition}

\theoremstyle{definition}

\theoremstyle{remark}
 
\newcommand{\R}{\mathbb{R}}

\newcommand{\dd}{\,\mathrm{d}}

\newcommand{\cL}{\mathcal{L}}
\newcommand{\cU}{\mathcal{U}}
\newcommand{\T}{\mathbb{T}}

\lstdefinestyle{mystyle}{
  language=Matlab,
  basicstyle=\ttfamily\small,
  commentstyle=\color{green!40!black},
  keywordstyle=\color{blue},
  numberstyle=\tiny\color{gray},
  numbers=left,
  stepnumber=1,
  backgroundcolor=\color{gray!5},
  frame=single,
  breaklines=true,
  breakatwhitespace=true,
  tabsize=2,
  morekeywords={optsSDE,default},
}

\title{Stability of stationary states for mean field models with multichromatic interaction potentials}

\author{Benedetta Bertoli$^1$, Benjamin D. Goddard$^2$, Grigorios A. Pavliotis$^1$}

\begin{document}

\maketitle

{\small \noindent $^1$Department of Mathematics, Imperial College London, London SW7 2AZ, UK.
\\
benedetta.bertoli21@imperial.ac.uk and pavl@ic.ac.uk
}

{\small \noindent $^2$School of Mathematics and the Maxwell Institute for Mathematical Sciences, University of Edinburgh, Edinburgh EH9 3FD, UK. b.goddard@ed.ac.uk}

\begin{abstract}

We consider weakly interacting diffusions on the torus, for multichromatic interaction potentials. We consider interaction potentials that are not $H$-stable, leading to phase transitions in the mean field limit. We show that the mean field dynamics can exhibit multipeak stationary states, where the number of peaks is related to the number of nonzero Fourier modes of the interaction. We also consider the effect of a confining potential on the structure of non-uniform steady states. We approach the problem by means of analysis, perturbation theory and numerical simulations for the interacting particle systems and the PDEs.

\end{abstract}

\section{Introduction}\label{sec:intro}

Nonlinear and nonlocal Fokker-Planck type equations, i.e. {McKean-Vlasov} equations appear in several applications, including stellar dynamics~\cite{BinneyTremaine2008}, plasma physics~\cite{balescu97}, mathematical biology~\cite{Painter_al_2024, primi2009mass}, active matter~\cite{peruani2008mean}, biophysics~\cite{frank04}[Sec 5.3], nematic liquid crystals~\cite{Constantin_al_2004} and models for opinion formation~\cite{Goddard_al_2021}. Such PDEs can exhibit nontrivial, i.e. non-uniform, stationary states, describing collective behaviour and the emergence of coherent structures, as an effect of interactions between agents at the microscale. In recent years, great progress has been made towards the understanding of the emergence of such collective behaviour. The purpose of this paper is to study the creation and stability of multimodal/multipeak stationary states for %nonlinear, nonlocal 
{McKean-Vlasov} equations on the torus.

In this paper, we will consider %nonlinear, nonlocal 
{McKean-Vlasov} equations of the form
\begin{equation}\label{e:mckean_external}
\frac{\partial \rho}{\partial t} = \beta^{-1} \Delta \rho + \nabla \cdot (\nabla V \rho) + {\kappa} \nabla \cdot ( ( \nabla W \star \rho) \rho), \quad \rho(x,0) = \rho_0(x),
\end{equation}
on $\mathbb{T}:=[0, 2 \pi]$ with periodic boundary conditions.  Here $\rho$ denotes the density/distribution function, $\rho_0$ the initial condition, $\beta$ the inverse temperature, {$\kappa$ the interaction strength}, $V$ and $W$ the confining and (symmetric) interaction potentials, respectively,  
{and $\star$ the convolution operator.}
Examples of dynamics described by~\eqref{e:mckean_external} are the Haken-Kelso-Bunz model from biophysics~\cite{frank04} with $V(x) = -\alpha \cos(x) - \gamma \cos(2 x)$ and $W(x) = - \cos(x)$, where $\alpha, \, \gamma$ are constants. We also mention the XY ($O(2)$) model with an external magnetic field that was studied in~\cite{MGRGGP2020}, corresponding to $V(x) = -\alpha \cos(x)$ and $W(x) = - \cos(x)$, as well as the  the noisy Kuramoto/Brownian mean field model~\cite{bertini2010dynamical, chavanis2014brownian}, $V \equiv 0$ and $W(x) = -\cos(x)$ and the noisy Hegselmann-Krause model for opinion dynamics~\cite{Goddard_al_2021, garnier2017consensus}. Several additional examples can be found in~\cite[Sec. 6]{CGPS2020}. Very nice presentations of the nonlinear McKean-Vlasov equation on the torus from a theoretical physics perspective can be found in~\cite[Sec. 5.3]{frank04} and~\cite{chavanis2014brownian}.

As is well known~\cite{CGPS2020}, the McKean-Vlasov PDE has a gradient structure:
\begin{equation}\label{e:grad_flow}
\partial_t \rho = \nabla \cdot \left(\rho \nabla \frac{\delta \mathcal{F}}{\delta \rho} \right),
\end{equation}
where $\mathcal{F}$ denotes the free energy
\begin{align}\label{e:free_en}
    \mathcal{F}(\rho, {\beta, \kappa}) = \beta^{-1} \int_{\mathbb{T}} \rho(x) \log (\rho(x)) \dd x + \int_{\mathbb{T}} V(x) \rho(x) \dd x + \frac{{\kappa}}{2} \int_{\mathbb{T}^d} \int_{\mathbb{T}} W(x-y) \rho(x) \rho(y) \dd x \dd y. 
\end{align}
Stationary states of the mean field dynamics can be characterized as critical points of the free energy~\cite{CGPS2020}. The main goal of this paper is to study the dynamical stability of  such states, in particular of stationary states that describe collective, organized behaviour.

Stationary states of the McKean-Vlasov PDE satisfy the Kirkwood-Monroe/generalized Lane-Emden integral equation~\cite{Bavaud_1991, CGPS2020, frank04, LuVuk_2010}
\begin{equation}\label{e:lane_emden}
\rho = \frac{1}{Z} e^{-\beta (V + {\kappa} W \star \rho)}, \quad Z = \int_{\mathbb{T}} e^{-\beta (V + {\kappa} W \star \rho)} \, \dd x.
\end{equation}
In the absence of an external potential, the uniform distribution, describing the disordered state, is always a stationary state of the {McKean-Vlasov} equation~\eqref{e:mckean_external}. Collective, organized behaviour, described by localized or multipeak solutions, becomes possible when the disordered state becomes unstable. 

As is well-known, the {McKean-Vlasov} equation~\eqref{e:mckean_external}, arises in the mean field limit of a system of weakly interacting diffusions~\cite{CDP2020, oelschlager1984martingale, chavanis2014brownian, MartzelAslangul2001}. In particular, we consider a system of interacting diffusions of the form: 
\begin{equation}\label{e:IPS}
    \dd x_i(t) = -\nabla V(x_i) \dd t -\frac{{\kappa}}{N} \sum_{j=1}^N \nabla W(x_i - x_j) \dd t + \sqrt{2 \beta^{-1}} \dd B_i (t), \quad x_i(0) \sim \rho_0,
\end{equation}
for $1 \leq i \leq N$, where $B_i(t)$ denote standard, 
%$d$-dimensional 
independent Brownian motions.
%Here $K \in \R$ is the interaction strength are independent Brownian motions. As it only suffices to vary one of the parameters $K$ or $\beta$, from now on we set $K =1$. $W:[0,2\pi] \rightarrow \R$ is the interaction potential. 
{For simplicity, we have assumed chaotic initial conditions, i.e. the $N$ particles are independent and identically distributed at time $t = 0$.}
Under appropriate assumptions on the confining and interaction potentials, the sequence of empirical measures $\rho^N := \frac{1}{N} \sum_{i=1}^N \delta_{x_i(t)} $ converges to the solution of the mean field PDE~\eqref{e:mckean_external} {as $N \rightarrow \infty$}. Equivalently, the $N$-particle distribution function for the interacting particle system~\eqref{e:IPS} can be written as $\rho^N(x_1,\dots x_N,t) \approx \Pi_{j=1}^N \rho(x_j,t)$. Rigorous convergence results, either at the level of the empirical measure or of the product measure structure of the $N$-particle distribution function, are by now well-estabilished and we refer to, e.g.,~\cite{oelschlager1984martingale, Lacker_2023, CDP2020}. %We will refer to this equation as the McKean-Vlasov PDE.

%\begin{itemize}

%\item Stability of stationary states.

%\item multipeak solutions

%\item Goal: repeat the detailed analysis that has been done for the Brownian mean field model to more general interaction and confining potentials.

%\end{itemize}

\subsection{Literature Review}

There is extensive literature on the calculation of stationary states for the McKean-Vlasov PDE and on the study of their stability as well as on applications to mathematical biology, in particular mass-selection in alignment models with non-deterministic effects and in active matter. The number of stationary states of the McKean-Vlasov PDE and their stability has been studied extensively, either by studying the Kirkwood-Monroe map~\cite{Bavaud_1991}, or by studying critical points of the free energy functional~\cite{LuVuk_2010}, or by studying the stationary McKean-Vlasov PDE.\footnote{All these approaches are, of course, equivalent; see~\cite[Prop. 2.4]{CGPS2020}.} The existence and stability of multipeak solutions, the problem that we will primarily focus on in this paper, was studied in \cite{primi2009mass, geigant2012stability} using PDEs/ODEs techniques. In \cite{primi2009mass}, the existence of a 2-peak steady state is proved under appropriate assumptions on the interaction potential. The stability of steady states was investigated numerically; the simulations presented in {\cite{primi2009mass}} suggest that only $1$ and $2$-peak steady states can be stable, while solutions with $4$ peaks are always unstable. The work in~\cite{geigant2012stability} builds on these results. 
%it is proved that at large enough, the constant stationary solution is locally stable; their value for the critical temperature is achieved by finding the eigenvalues of the linearisation of the McKean-Vlasov equation. Their method is slightly different from ours, but the value of the critical temperature agrees with ours. They then discuss stability of steady states and use numerical simulations to analyse specific examples of interaction potentials. 
Bifurcation theory for the stationary McKean-Vlasov equation for multichromatic interaction potentials was recently analysed in \cite{vukadinovic2023phase}, {where} the Hodgkin-Huxley oscillator model is studied, with an interaction potential consisting of two Fourier modes with opposite sign, $W(\theta) = - \cos(\theta) + \varepsilon \cos(2(\theta-\alpha))$, $\varepsilon \geq 0, \alpha \in (-\pi/2, \pi/2)$. 
%using properties of modified Bessel-functions. 
Critical points of the free energy functional for the Onsager model for liquid crystals, $W(x) =  |\sin(x)|$ have been studied extensively.  See, e.g.~\cite{LuVuk_2010, FS_2005a, FS_2005b} and the references therein.
%{\bf To do: add references to the Onsager model, Smoluchowski equation, biological applications}.

A quite comprehensive theory of bifurcations from the uniform distribution and of phase transitions for the McKean-Vlasov PDE on the torus was developed in~\cite{CGPS2020}. The goal of the present paper is to study, by means of analysis, systematic perturbation theory, and numerical simulations, the stability of non-uniform states, and in particular of multipeak solutions. The stability of the non-uniform state for the noisy Kuramoto model was studied using spectral theoretic arguments in~\cite{bertini2010dynamical}. One of the goals of the present study is to extend the analysis from this paper to multichromatic interaction potentials.

% The stability of multipeak solutions has been studied in several publications, in particular with a view towards applications to mathematical models in biology. In~\cite{primi2009mass} conditions on the interaction potential for the existence and stability of multipeak solutions were studied. The stability of multipeak solutions, both in the presence and the absence of diffusion, was studied in~\cite{geigant2012stability}.

\subsection{Our Contributions}

In this paper we consider the stability of multipeak stationary solutions for the one-dimensional McKean-Vlasov PDE, both in the presence or absence of a confining potential. Our main contributions are the following:

\begin{itemize}

\item We provide a detailed stability analysis of the non-uniform state for multichromatic potentials.

\item In particular, we find that the uniform state changes stability at some critical value of the temperature $\beta$, and that multi-peak stationary states are unstable.

\item We calculate the eigenvalues of the linearized McKean-Vlasov operator above the {transition point}.

\item We present very detailed numerical experiments by solving the evolution PDE, the SDEs for the interacting particle system, and the eigenvalue problem for the linearized operator.

\end{itemize}

The rest of the paper is organized as follows. In Section~\ref{sec:SelfConsistency} we present the models that we will consider in this paper and we analyse the self-consistency equation(s). In Section~\ref{sec:stability_analysis} we study the stability of stationary states by either calculating the second variation of the free energy or by linearizing the McKean-Vlasov PDE. {Perturbative} results for the eigenvalues of the linearized McKean-Vlasov operator close to the bifurcation point are presented in Section~\ref{sec:perturb}. The results of extensive numerical simulations based on both the PDE and SDE formulations are shown in Section~\ref{sec:num_exp}. Conclusions and comments on future work are presented in Section~\ref{sec:conclusions}.

%Models to consider:

%\begin{itemize}
%    \item Multichromatic interaction potential, start with $W(x) = -A \cos(s) - B \cos(2x)$.
%    \item Consider alsot $W(x) = -|cos(x)|^{\ell}$.
%    \item Confining potential $V(x) = A \cos(x) + B \cos(2x)$ and $W(x) =  - \cos(x) - \gamma \cos(2x)$.
%    \item Main question: stability of multipeak solutions. Relation between the number of Fourier modes (critical points) of the interaction potential and the number of peaks of the stationary state(s). 
%    \item Consider the overdamped, underdamped and "active matter" dynamics. What is the effect of the parameters in these three types of dynamics? Do linear stability analysis.
%    \item Study of fluctuations.
%    \item Monte Carlo simulations.
%    \item Solution of the integral equation for the stationary states.
%\end{itemize}

\section{Set-up and self-consistency equations}
\label{sec:SelfConsistency}

\subsection{Phase transitions, stability analysis and the self-consistency equation}

For $H$-stable potentials, i.e. interaction potentials with non-negative Fourier coefficients, and in the absence of a confining potential, the free energy functional is convex and the uniform distribution is the unique, globally stable stationary state~\cite{Bavaud_1991}.

For monochromatic interaction potentials of the form $W(x) = -\cos(k x), k \in \mathbb{N}$, and in the absence of a confining potential, a detailed characterization of stationary states was given in~\cite{CGPS2020}. In particular, we have the following:

\begin{prop}~\label{prop:kura}
The generalised Kuramoto model $W(x)=-\cos(k x)$, for some $k \in \mathbb{N}$, $k \neq 0$ exhibits a continuous transition point at the linear instability threshold ${\beta}_c$. For {$\beta>\beta_c$}, the equation {$\mathcal{F}(\rho,\beta, \kappa)=0$} has only two solutions in {$L^2([0,2\pi])$} (up to translations). The nontrivial one, {$\rho_{\beta}$ minimises $\mathcal{F}_\beta$ for $\beta>\beta_c$ and converges in the narrow topology as $\beta \to \infty$} to a normalised linear sum of equally weighted Dirac measures centred at the minima of $W(x)$.
\end{prop}
%\textcolor{blue}{BEN: I think there is a subscript c missing on one of the $\beta$s in the red sentence above.}

We will consider even confining and interaction potentials with a finite number of non-zero Fourier modes:
\begin{align*}
    V(x) = \sum_{k=1}^{{n'}} v_k \cos(k x), \quad \mbox{and} \quad  W(x) =\sum_{k=1}^n a_k \cos(k x), 
\end{align*}
 %$a_1, \ldots, a_n$ are negative constants.
 with $\{ v_k\}_{k=1}^{{n'}}$ real-valued and $\{ a_k\}_{k=1}^n$ non-positive. {There are several applications, for example in biophysics, polymer dynamics, liquid crystals, synchronization and mathematical biology where the McKean-Vlasov PDE and the corresponding elliptic problem for the stationary state(s) is posed on the torus. These potentials are of interest for several reasons. Firstly, they provide a generalisation of the noisy Kuramoto model, which is useful in real-world applications (see, for example, \cite{PhysRevLett.131.158303}). From a mathematical perspective, as the McKean-Vlasov equation we consider is posed on the torus with periodic boundary conditions, it is natural to consider potentials which are periodic functions of $x$. Moreover, as the general potentials of interest in this type of equations are even functions, they can be expressed in terms of cosine Fourier series, of which $W$ and $V$ as above represent finite truncations.}

%Therefore, the McKean-Vlasov becomes:
%\begin{equation}\label{PDE}
%    \frac{\partial \rho}{\partial t} = -\frac{\partial}{\partial x} \left[ \rho \left(  \sum_{k=1}^m k v_k \sin(k x) +  \sum_{k=1}^n k a_k \int_0^{2\pi} \rho(t, y)  \sin(k(x-y)) \right) \right] + \beta^{-1} \frac{\partial^2 \rho}{\partial x^2}
%\end{equation}

The main objective of this paper is to study the nature and stability of stationary states of the McKean-Vlasov PDE by analysing the Kirkwood-Monroe integral equation~\eqref{e:lane_emden}. 
We reiterate that, for nontrivial confining potentials $V(x)$, the uniform state is no longer a stationary state. 
In particular, if $W(x) = \sum_{k=1}^n a_k \cos(kx)$, $V(x) = \sum_{k=1}^{{n'}} v_k \cos(kx)$, then {every stationary state can be written as}:
\begin{equation}\label{e:stationary_state}
    {\rho_{\infty}}(x) = \frac{1}{Z} \exp \left( -\beta \sum_{k=1}^{\max\{n, {n'}\}} (v_k + r_k a_k) \cos(kx) \right),
\end{equation}
where $Z = \int_0^{2\pi} \exp \left( \beta \sum_{k=1}^{\max\{{n'},n\}} (v_k + r_k a_k) \cos(kx) \right)$, and we set $v_k = 0$ for $k>{n'}$, and $a_k = 0$ for $k>n$.

The Fourier coefficients of $\rho$, $r_k = \int_0^{2\pi} \cos(kx) \rho(x) \dd x$, solve the self-consistency equations:
\begin{align*}
    r_l = \frac{1}{Z} \int_0^{2\pi} \cos(lx) \exp \left( -\beta \sum_{k=1}^{\max\{{n'},n\}} (v_k + r_k a_k) \cos(kx) \right) {\dd x},
\end{align*}
for $1 \leq l \leq \max\{{n'},n\}$

These self-consistency equations, in the absence of a confining potential, were first derived in~\cite{Battle1977}. {For $V=0$, $r_l = 0$ for all $1 \leq l \leq n$ is always a solution, corresponding to the uniform stationary state $\rho_{\infty} = \frac{1}{2\pi}$}.

{Beyond the uniform steady state}, we are interested in the number of critical points of the steady states of the McKean-Vlasov equation (\ref{e:mckean_external}).  When there is no external potential, it is conceivable that one may relate this to the number of Fourier modes of the interaction potential.  However, as we show below, even if one fixes $W$, then this is not possible for general external potentials.  As a particular example, one may fix $W$ and (smooth, non-negative) $\rho$ and then choose $V$ such that $\rho$ is a steady state, irrespective of the number of peaks it contains.
\begin{lem}
Let $\hat{\rho}$ be a smooth, non-negative, $2\pi$-periodic function and fix the interaction kernel $W$.  Then there exists an external potential, $V$, such that $\hat{\rho}$ is an equilibrium of the McKean-Vlasov equation (\ref{e:mckean_external}).
\end{lem}
\begin{proof}
By the non-negativity of $\hat{\rho}$, we may write 
$\hat{\rho}(x) = \exp\big(- \beta f(x)\big)$ for some $f(x)$, which is fixed by the choice of $\hat{\rho}$.
Similarly, by (\ref{e:lane_emden}), if 
$\hat{\rho}(x) = \exp\big[ -\beta \big( V + W\star {\hat\rho} - \ln(Z) \big)  \big]$, then $\hat{\rho}$ is an  equilibrium of (\ref{e:mckean_external}).  Note that, for convenience, we have written the normalization constant in the exponent.
Hence it is clear that we want $f(x) = V(x) + (W\star {\hat{\rho}})(x) - \ln(Z)$, which can be achieved by taking $V(x) = f(x) + (W\star{\hat{\rho}})(x) - c$, where $c$ is a constant chosen to ensure the correct normalization.
\end{proof}

We note that, in fact, the maximum number of non-zero Fourier modes required to produce an equilibrium distribution with $M$ non-zero Fourier modes under a kernel with $n$ non-zero Fourier modes is $\max\{n,M \}$. 
{\subsection{Examples}}
{We now present some examples of particle systems and their corresponding self-consistency equations, all of which belong to the class of equations studied here.}

{ \bf Kuramoto model (monochromatic interactions)}
Firstly, there are many well-known results for the interaction potential $W(x) = - \cos(x)$ {(and $V(x)\equiv 0$)}, which gives rise to the Kuramoto model.  In particular, it is known (see \cite{bertini2010dynamical}) that the McKean-Vlasov PDE for the Kuramoto model undergoes a phase transition. This means that there exists a critical value $\beta_c$ of the inverse temperature $\beta$ such that, for $\beta < \beta_c$, the PDE admits a unique stationary solution (the uniform state $1/(2\pi)$), while for $\beta > \beta_c$, the PDE has multiple stationary states. Each of these stationary solutions can be written as $q(x) = \rho(x + x_0)$ for some $x_0 \in (0,2\pi)$, where 
\begin{align*}
    \rho(x) = \frac{1}{Z} \exp(-\beta r \cos(x)),
\end{align*}    
with $Z = \int_0^{2\pi}\exp(-\beta r \cos(y)) \dd y$. Here $r$ is a solution of the self-consistency equation:
\begin{align}\label{e:self_consist_single}
    r = \frac{\int_0^{2\pi} \cos(x) \exp(-\beta r \cos(x)) \dd x}{\int_0^{2\pi} \exp(-\beta r \cos(x)) \dd x}.
\end{align}

{ \bf The noisy Kuramoto model with an external potential} As an example in which the self-consistency equation can be solved analytically, we consider the following Brownian mean field model in a magnetic field studied in~\cite{MGRGGP2020}, for the confining and interaction potentials $V(x) =-\eta \cos(x)$, $W(x) = -a \cos(x)$. This can be written as an unbounded spin system with a Hamiltonian, where for ${\bf s}_i = (\cos(x), \sin(x))$
\begin{equation}\label{e:BMF_magnetic}
H = -\sum_{i=1}^N {\bf J} \cdot {\bf s}_i - \frac{2}{N} \sum_{i,j=1}^N {\bf s}_i \cdot {\bf s}_j, 
\end{equation}
with ${\bf J} = - \eta (1,0)$. In this case we can calculate the stationary states analytically~\cite{MGRGGP2020}: for $\beta < \beta_c$ there exists a unique steady state
\begin{equation}\label{e:inv_meas_magnetic}
\rho_{min}(x) = \frac{1}{Z_{min}} e^{a^{min} \cos(2 \pi x)}, \quad Z_{min} = \int_{\T} e^{a^{min} \cos(2 \pi x)} {\dd x},
\end{equation}

Furthermore, for $\beta > \beta_c$ there exist at least two steady states; in addition to~\eqref{e:inv_meas_magnetic}, we have
\begin{equation}\label{e:inv_meas_magnetic_beta}
\rho_{*}(x) = \frac{1}{Z_{min}} e^{a^{*} \cos(2 \pi x)}, \quad Z_{*} = \int_{\T} e^{a^{*} \cos(2 \pi x)} {\dd x},
\end{equation}
for some $a^{*} = a^{*}(\beta)$, with $a^* < 0 < a^{min}$. Here $\rho^{min}$ is the unique minimizer and $\rho^*$ is a non-minimizing critical point  of the periodic mean field energy. 

{ \bf Multichromatic interactions} Our aim is to extend these types of results to more general interaction potentials $W$ with a varying number of Fourier modes, such as the bichromatic potential $W(x) = a_1 \cos(x) + a_2 \cos(2x)$, $a_1, a_2 < 0$. We address the question of existence and uniqueness of stationary states and their stability.
{By \eqref{e:stationary_state}}, invariant measures for the corresponding McKean-Vlasov PDE are given by:
\begin{equation}\label{stationaryrho}
  \rho(x) = \frac{\exp(-\beta(a_1 r_1 \cos(x) + a_2 r_2 \cos(2x)))}{\int_0^{2\pi}\exp(-\beta(a_1 r_1 \cos(y) + a_2 r_2 \cos(2y))) \dd y},
\end{equation}
with $r_1$ and $r_2$ satisfying the self-consistency equations:
\begin{align*}
  r_1 = \frac{\int_0^{2\pi} \cos(x)\exp(-\beta(a_1 r_1 \cos(x) + a_2 r_2 \cos(2x))) \dd x}{\int_0^{2\pi} \exp(-\beta(a_1 r_1 \cos(y) + a_2 r_2 \cos(2y))) \dd y},
\end{align*}
and
\begin{align*}
  r_2 = \frac{\int_0^{2\pi} \cos(2x)\exp(-\beta(a_1 r_1 \cos(x) + a_2 r_2 \cos(2x))) \dd x}{\int_0^{2\pi} \exp(-\beta( a_1 r_1 \cos(y) + a_2 r_2 \cos(2y))) \dd y}.
\end{align*}

In this case, it is harder to rigorously deduce results about $r_1$ and $r_2$. However, as this potential only induces two self-consistency equations, we can solve these numerically. We take, for example, $a_1 = a_2 = -1$.
It is easy to notice that for any value of $\beta$, $r_1 = r_2 = 0$ is a solution, which corresponds to the uniform stationary state $\rho = 1/(2\pi)$. Numerical experiments indicate that for $\beta < 2$, this is the only solution. However, when $\beta > 2$, we start seeing two other pairs of solutions $(r_1, r_2)$; substituting these into (\ref{stationaryrho}) gives us three different steady states. In Figure~\ref{fig:steadystates_beta} we show these steady states for for $\beta = 3$ and $\beta = 10$, along with the corresponding values for $r_1$ and $r_2$.
\begin{figure}[H]
  \centering
  \begin{minipage}[b]{0.49\textwidth}
    \includegraphics[width=\textwidth]{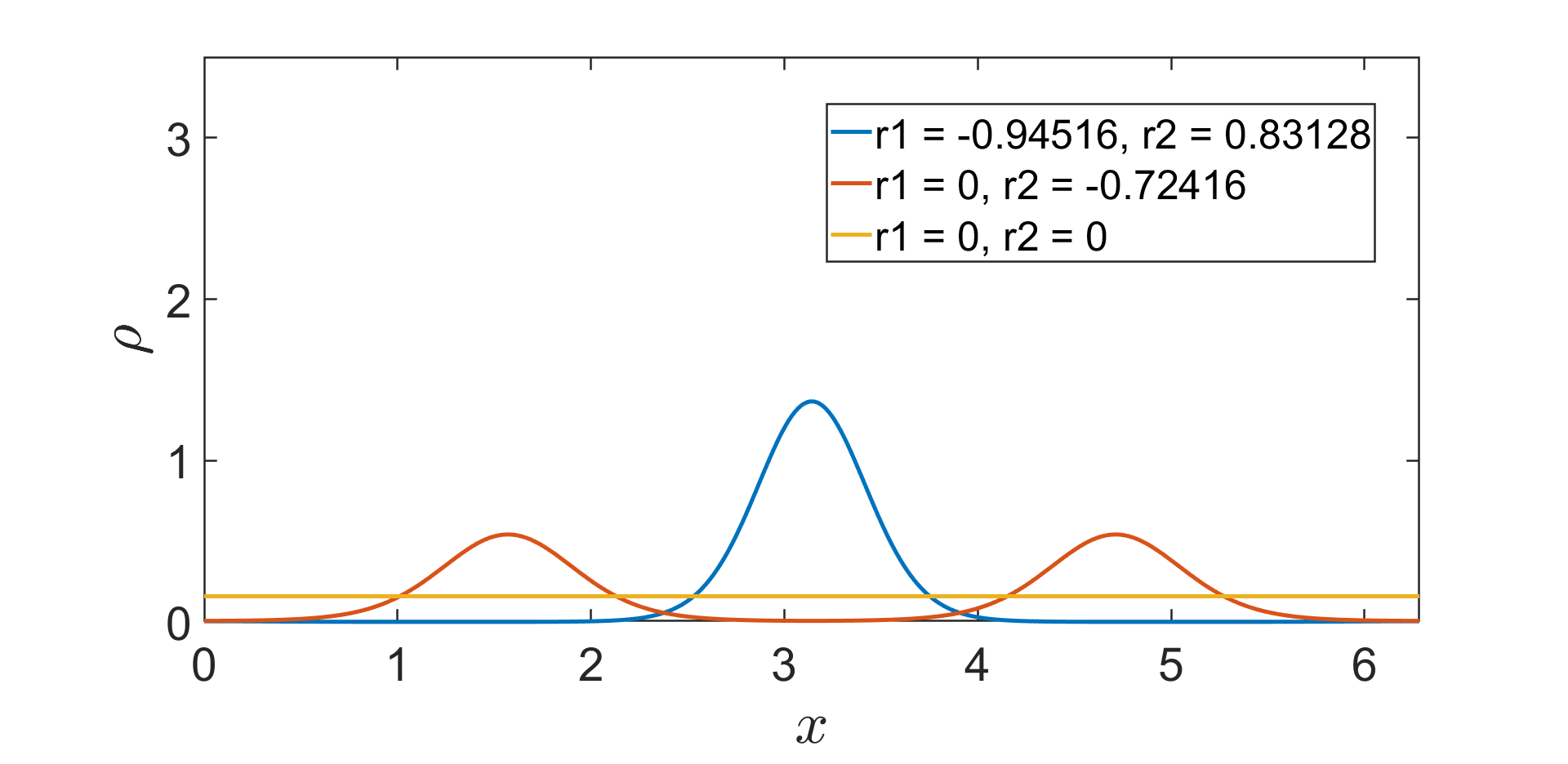}
  \end{minipage}
  \hfill
  \begin{minipage}[b]{0.49\textwidth}
    \includegraphics[width=\textwidth]{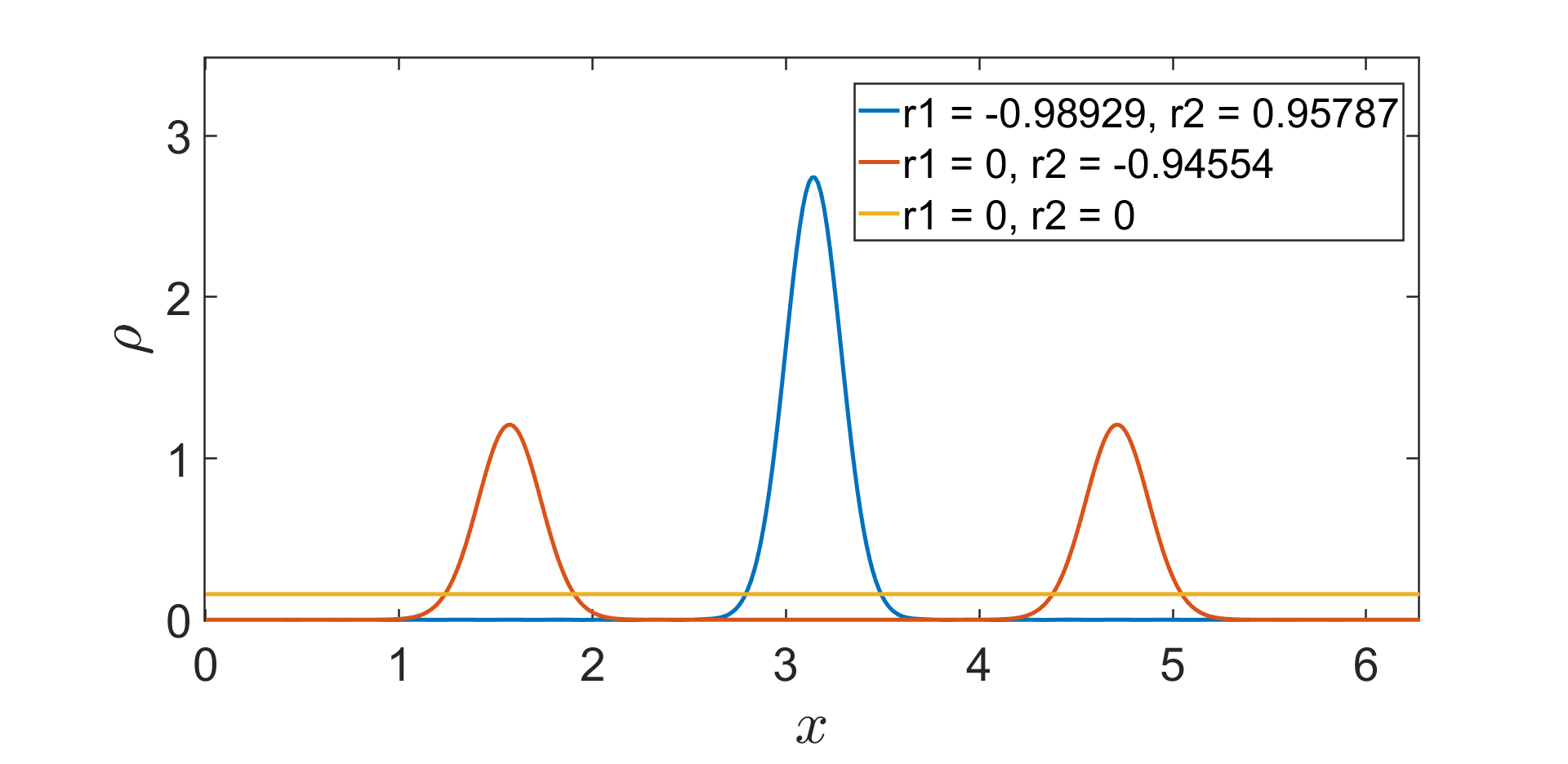}
  \end{minipage}
    \caption{Probability distributions corresponding to the three different solutions {of the $(r_1, r_2)$ equations for the bichromatic system with} $\beta = 3$ (left) and $\beta = 10$ (right).}
    \label{fig:steadystates_beta}
\end{figure}
Note that, alongside the uniform distribution, there is a one-peak steady state (as also seen for the Kuramoto model), and we now have a further solution with two peaks, reflecting the multichromatic nature of the interaction potential.

In the following section we aim to address the question of stability of such solutions.

\section{Stability analysis of the uniform state}
\label{sec:stability_analysis}

{The goal of this section is to evaluate the critical temperature at which the systems presented above undergo a phase transition. Here, we assume that there is no confining potential, i.e., $V=0$. In this case, the McKean-Vlasov equation \eqref{e:mckean_external} simplifies to:
\begin{align*}
    \frac{\partial \rho}{\partial t} = \beta^{-1} \Delta \rho + \kappa \nabla \cdot ( ( \nabla W \star \rho) \rho), \quad \rho(x,0) = \rho_0(x),
\end{align*}
As one of the parameters $\kappa, \beta$ is redundant when considering steady states, we also set the interaction strength $\kappa = 1$ whenever $V=0$, and focus on finding the phase transition in terms of $\beta$ only. As before, for interaction potentials that are not $H-$stable we expect a critical value of $\beta$ above which the McKean-Vlasov PDE admits stationary solutions other than the uniform state. The key result is that the critical value of $\beta$ is:
\begin{align*}
\beta_c = \min_{1 \leq k \leq n} \left\{ -\frac{2}{a_k} : 1 \leq k \leq n \right\}.
\end{align*}
In particular, the critical value is not influenced by the number of Fourier modes of the interaction potential, nor by which Fourier mode has the highest coefficient, but rather by the magnitude of the coefficient itself.}
We will identify this critical temperature in three ways: firstly we analyse the second variation of the free energy, then we perform a linear stability analysis of the PDE, and finally we analyse the problem numerically. These methods will also give us an insight on the stability of the states we find.

\subsection{Second variation of the free energy}
\label{subsec_free_energy}

The first approach we use to perform the linear stability analysis and to identify the continuous phase transition, based on the one used in \cite{chavanis2014brownian}, consists in analysing the eigenvalues of the second variation of the free energy. We first present the calculation in the absence of a confining potential. We recall that the free energy of this system is given by
\begin{align*}
    \mathcal{F}(\rho) = \frac{1}{2} \int_0^{2\pi} \int_0^{2\pi} W(x-y) \rho(x) \rho(y) \dd x \dd y + \beta^{-1} \int_0^{2\pi} \rho(x) \log (\rho(x)) \dd x .
\end{align*}
Its second variation is
\begin{equation}\label{e:second_variation_energy}
    \delta^2 \mathcal{F}[\delta \rho, \, \delta \rho] = \frac{1}{2} \int \delta \rho \delta \phi \dd x + \frac{\beta^{-1}}{2} \int \frac{(\delta \rho)^2}{\rho} \dd x ,
\end{equation}
where $\delta \rho$ is the perturbation and
\begin{align*}
    \delta \phi(x) = \int_0^{2\pi} W(x - y) \delta \rho (y) \dd y.
\end{align*}
%\textcolor{blue}{BEN: we don't have the minus sign in the definition of $W$ above - I think we should be consistent}
For our interaction potential $W(x) = \sum_{k=1}^n a_k \cos(kx)$, this is equal to
\begin{align*}
    \delta^2 \mathcal{F} = (2\pi)^2 \sum_{j=1}^{\infty} \left( \beta^{-1} + \frac{a_j}{2} \right) |\delta \hat{\rho}_j|^2 ,
\end{align*}
{where $\hat{\rho}_j$ is the $j$-th Fourier coefficient of $\rho$}.
We then write
\begin{align*}
    \delta \rho = \frac{\dd q}{\dd x}, \text{ where } q(x) = \int_0^{x} \delta \rho (y, t) \dd y.
\end{align*}
% to obtain:
% \begin{align*}
%     \delta^2F &= \frac{1}{2} \int_0^{2\pi} \delta \rho \delta \phi \dd x + \frac{\beta^{-1}}{2} \int_0^{2\pi} \frac{(\delta \rho)^2}{\rho} \dd x .
% \end{align*}
This gives, for the first integral {in \eqref{e:second_variation_energy}}:
\begin{align*}
    &\frac{1}{2} \int_0^{2\pi} \delta \rho \delta \phi \dd x = 
    %\frac{1}{2} \int_0^{2\pi} \frac{\dd q}{\dd x} \int_0^{2\pi} W(x - y) \frac{\dd q}{\dd y} \dd y \dd x \\
     \frac{1}{2} \int_0^{2\pi} \int_0^{2\pi} q(x) \frac{\dd}{\dd x} \frac{\dd}{\dd y} W(x - y) q(y) \dd x \dd y,
\end{align*}

and for the second integral:
\begin{align*}
    \frac{\beta^{-1}}{2} \int_0^{2\pi} \frac{(\delta \rho)^2}{\rho} \dd x 
    %&= -\frac{\beta^{-1}}{2} \int_0^{2\pi} q(x) \frac{\dd}{\dd x} \left( \frac{1}{\rho(x)} \delta \rho(x) \right)  \dd x \\
    &= -\frac{\beta^{-1}}{2} \int_0^{2\pi} q(x) \frac{\dd}{\dd x} \left( \frac{1}{\rho(x)} \frac{\dd}{\dd x} q(x) \right)  \dd x.
\end{align*}

% Moreover
% \begin{align*}
%     \frac{\dd}{\dd x} \frac{\dd}{\dd y}W(x - y) = - \sum_{k=1}^n k^2 a_k \cos(k(x - y))
% \end{align*}

Therefore, the second variation of the free energy can be written as a quadratic form/integral operator:
\begin{align*}
    \delta^2 \mathcal{F} [q, q] = \int_0^{2\pi} \int_0^{2\pi} K(x, y) q(x) q(y) \dd x \dd y,
\end{align*}
where the kernel is defined as 
$$
K(x, y) = \frac{1}{2} \sum_{k=1}^n k^2 a_k \cos(k(x-y)) - \frac{\beta^{-1}}{2} \delta(x - y) \frac{\dd}{\dd x} \left( \frac{1}{\rho} \frac{\dd}{\dd x} \right).
$$
We are therefore led to considering the eigenvalue problem:
\begin{equation}\label{evalue}
  2\pi \frac{\dd^2 q}{\dd x^2}  + \beta \int_0^{2\pi} q(y) \left( \sum_{k=1}^n k^2 a_k \cos(k(x - y)) \right) \dd y= -2\lambda q.
\end{equation}
%Our goal is to calculate the eigenvalues  $\lambda$ for 
The eigenmodes are {$q_s = A_s \sin(sx)$, where $A_s\in \mathbb{R}, s \in \mathbb{N}$}. 
Using trigonometric identities, we can compute the eigenvalues of the integral operator. We first conclude that, for $ n < |s|$, we have $\lambda_s = \pi s^2 >0$, so these modes do not induce instability. For $|s| \leq n$, we have 
%the expression equals:
%\begin{align*}
%  &-2\pi m^2 A_m \cos(mx) - \beta A_m \cos(mx) m^2 a_m \int_0^{2\pi} (\cos(my))^2 \dd y \\
%  &= A_m \cos(mx) \left(-2\pi m^2 - \beta m^2 a_m \pi \right).
%\end{align*}
%Setting this equal to $-2\lambda q$ gives that:
\begin{align*}
  \lambda_s = - \frac{1}{2} \left( -2\pi s^2 - \beta s^2 a_s \pi \right) = {\frac{\pi s^2}{2}} \left( 2 + \beta a_s \right)
\end{align*}
This is positive for $\beta a_s < -2$, so $\beta < -\frac{a_s}{2}$. As we are concerned with the first point of linear instability, the critical temperature is $\beta_c := \min_k \{ -\frac{2}{a_k} : 1 \leq k \leq n \}$, see also~\cite{CGPS2020}. As an example, we note that by setting $a_1 = -1, a_k = 0$ for $2 \leq n$, we obtain the known result for the Kuramoto model, $\beta_c = 2$. Similarly, for the interaction potentials $W(x) = -3\cos(x) - \cos(2x)$ or $W(x) = -\cos(x) - \cos(2x) - 3\cos(3x)$, the critical temperature is $\beta_c = \frac{2}{3}$, as the critical value is not influenced by which Fourier mode has the highest coefficient, but rather by the magnitude of the coefficient itself. Critical temperatures for more examples can be found, together with the corresponding numerical simulations, in Section \ref{sec:num_exp}.

\subsection{Linearisation of the McKean-Vlasov equation}
\label{subsec:}

We can calculate the value of the critical temperature also by looking at the Fourier modes of the linearisation of the {McKean-Vlasov} PDE
\begin{align*}
    \frac{\partial \rho}{\partial t} (t,x) = \frac{\partial}{\partial x} \left[ \left( \int \rho(t,x-y)  W'(y) \dd y \right) \rho(t,x) \right] + \beta^{-1} \frac{\partial^2 \rho}{\partial x^2} (t,x).
\end{align*}

Following the method in \cite{garnier2017consensus}, we decompose $\rho = \rho_0 + \rho_1 = \frac{1}{2\pi} + \rho_1$, where $\rho_1$ is a small perturbation of $\rho$ so that $O(\rho_1^2)$ is negligible. In the Fourier domain the PDE becomes
\begin{align*}
    \frac{\partial \hat{\rho}_1(t,j)}{\partial t} = \left[ i \rho_0 j \int e^{-ijy}  W'(y) \dd y - \beta^{-1} j^2 \right] \hat{\rho_1}(t,j),
\end{align*}
where $\hat{\rho}_1(j)$ is the $j$-th Fourier coefficient of $\rho_1$. 
Therefore, the Fourier modes have growth rates
\begin{align*}
    \gamma_j = \text{Re} \left[ i \rho_0 j \int e^{-ijy} W'(y) \dd y - \beta^{-1} j^2 \right] = \rho_0 j \int \sin(jy)  W'(y) \dd y - \beta^{-1} j^2.
\end{align*}

Substituting $ W'(x) = -\sum_{k=1}^n k a_k \sin(kx)$, our problem reduces to studying the sign of:
\begin{align*}
\gamma_j =   -\frac{j}{2\pi} \int_0^{2\pi} \sin(jy) \sum_{k=1}^n k a_k \sin(ky) \dd y - \beta^{-1} j^2,
\end{align*}
as $\beta$ varies. For $j > n$, the integral term above is equal to $0$. Hence:
\begin{align*}
  \gamma_j = -\beta^{-1} j^2.
\end{align*}
As this is always negative, these modes do not induce instability. For $j \leq n$:
\begin{align*}
    \gamma_j = -\frac{1}{2\pi} \int_0^{2\pi} \sin(jy) \sum_{k=1}^n k a_k \sin(ky) \dd y - \beta^{-1} j  = -\frac{j a_j}{2} - \beta^{-1} j.
    %&= -\frac{1}{2\pi} \int_0^{2\pi} j a_j (\sin(jy))^2 \dd y  - \beta^{-1} j
\end{align*}

This is positive for $2 \beta^{-1} <-a_j$, i.e. for $\beta > -\frac{2}{a_j}$, as expected from the calculations in the previous sections.

\section{Stability analysis of the peaked states}
\label{sec:perturb}

%\subsection{Reduction to a Schr\"{o}dinger operator}
{Our next goal is to study the stability of non-uniform stationary states close to the critical interaction strength. One approach would be to linearize the McKean-Vlasov operator around the (non-uniform) stationary state $\rho_{\infty}$:
\begin{equation}\label{e:McK_Vl}
\cL \rho = \beta^{-1} \partial_x^2 \rho + \partial_x \left((\partial_x W \star \rho_{\infty}) \rho \right)
\end{equation}
on $[0,2 \pi]$ with periodic boundary conditions. We emphasize the fact that by linearization we mean that we consider the Fokker-Planck operator of the linear (in the sense of McKean) process that we obtain by calculating the convolution with respect to the invariant measure, see \cite{pavliotis2025linearization} for details. Alternatively, we can consider the McKean-Vlasov PDE as a "linear" PDE with a time dependent potential, transform it to a reaction-diffusion equation and then calculate the expectation in the "reaction" term with respect to the invariant measure, as done in \cite{vukadinovic2008inertial} (under the assumption of even initial conditions). 
%\textcolor{blue}{BEN: (14) contains an external potential, but (15) doesn't -- from what we do below I assume that we want $V \equiv 0$?}
To be consistent with the notation in existing literature, in particular with \cite{bertini2010dynamical}, we multiply this equation by $\beta$.

% The linearized operator is
% \begin{eqnarray}
% \cL_{\rho_{\infty}} \rho 
%    & = \beta^{-1} \partial_x^2 \rho + \partial_x ((\partial_x W \star \rho_{\infty}) \rho) +  \partial_x ((\partial_x W \star \rho) \rho_{\infty})
%  \nonumber \\ & =: & 
% \beta^{-1} \partial_x^2 \rho + \partial_x (\mathcal{U} \rho) + \partial_x  \label{e:schrod_nonlocal}
% \end{eqnarray}
% In writing the above, we have introduced the potential $\mathcal{U} := W \star \rho_{\infty}$. This is an integro-differential Fokker-Planck operator. Via the standard ground state transformation we can map it to a nonlocal Schr\"{o}dinger operator~\cite{Pavl2014}[Sec. 4.9], of the form considered in~\cite{Davidson_Dodds_2006, Davidson_Dodds_2006b}, albeit with different boundary conditions. 
We follow \cite{bertini2010dynamical}[Sec. 2.5] (Eqn (2.57)) and perform the ground state transformation before the linearization to map the McKean-Vlasov operator to a nonlinear and nonlocal Schr\"{o}dinger operator. The nonlinearity and nonlocality enter
through the dependence on the Fourier modes of the stationary states that satisfy the
self-consistency equations. See also \cite{vukadinovic_2009}[Sec. 3].
The transformation is achieved by considering the rescaling $Hf := \rho_{\infty}^{-1/2} \cL(\rho_{\infty}^{1/2} f)$. 
We recall that in this case, the stationary state $\rho_{\infty}$ is given by:
\begin{align*}
    \rho_{\infty} = \frac{1}{Z} e^{-\beta \cU}
\end{align*}
where $Z$ is the renormalisation constant and $\cU := W \star \rho_{\infty}$.
We have:
\begin{align*}
    &\cL(\rho_{\infty}^{1/2} f) = \partial_{x}^2 \left( e^{-\frac{\beta \cU}{2}} f \right) + \beta \partial_x \left( (\partial_x \cU) e^{-\frac{\beta \cU}{2}} f \right) \\
    &= \partial_x \left( - \frac{\beta}{2} e^{-\frac{\beta \cU}{2}} (\partial_x \cU) f + e^{-\frac{\beta \cU}{2}} \partial_x f \right) + \beta \left( \partial_x^2 \cU e^{-\frac{\beta \cU}{2}} f - \frac{\beta}{2} | \partial_x \cU |^2 e^{-\frac{\beta \cU}{2}} + \partial_x \cU e^{-\frac{\beta \cU}{2}} \partial_x f \right) \\
    %&=\left( \frac{\beta^2}{4} | \partial_x \cU |^2 e^{-\frac{\beta \cU}{2}} f - \frac{\beta}{2} e^{-\frac{\beta \cU}{2}} \partial_x^2 \cU f - \frac{\beta}{2} e^{-\frac{\beta \cU}{2}}\partial_x \cU \partial_x f - \frac{\beta}{2} \partial_x \cU e^{-\frac{\beta \cU}{2}} \partial_x f + e^{-\frac{\beta \cU}{2}} \partial_x^2 f \right) \\ 
    %&+ \beta \left( \partial_x^2 \cU e^{-\frac{\beta \cU}{2}} f - \frac{\beta}{2} | \partial_x \cU |^2 e^{-\frac{\beta \cU}{2}} f + \partial_x \cU \partial_x f e^{-\frac{\beta \cU}{2}} \right) \\
    &= e^{-\frac{\beta \cU}{2}} \left( - \frac{\beta^2}{4} |\partial_x \cU |^2 f + \frac{\beta}{2} (\partial_x^2 \cU) f + \partial_x^2 f \right).
\end{align*}
%\textcolor{blue}{BEN: have we defined $\cU$?}
Therefore
\begin{align*}
    Hf = \partial_x^2 f - \Phi f, 
\end{align*}
where
\begin{align*}
    \Phi f = \left( \frac{\beta^2}{4} | \partial_x \cU |^2 - \frac{\beta}{2} \partial_x^2 \cU \right) f.
\end{align*}
We aim to find the eigenvalues of $H$ near the critical value of $\beta$, as this will give us further information about what happens in the dynamics near the transition point. In particular, we show that all eigenvalues are negative, which implies that, close to the transition point, the system stays close to the steady state in the short term. Moreover, the main eigenvalue of interest is the largest one, as it gives us information about the spectral gap of the operator.
}
{\subsection{The Kuramoto model}}
{For the Kuramoto model, $\rho_{\infty} = \frac{1}{Z} e^{\beta r \cos(x)}$, where $Z$ is the renormalization factor. In this case:
\begin{align*}
    \cU &= -\frac{1}{Z} \int_0^{2\pi} \cos(x-y) e^{\beta r \cos(y)} \dd y 
    = -\frac{1}{Z} \cos(x) \int_0^{2\pi} \cos(y) e^{\beta r \cos(y)} \dd y \\
    &= -r \cos(x).
\end{align*}
Therefore
\begin{align*}
    Hf &= \partial_x^2 f - \left( \frac{\beta^2}{4} |\partial_x \cU|^2 - \frac{\beta}{2} \partial_x^2 \cU \right) f 
    = \partial_x^2 f - \left( \frac{\beta^2}{4} (r\sin(x))^2 - \frac{\beta}{2} (r\cos(x)) \right) f.
\end{align*}
We consider the self-consistency equation~\eqref{e:self_consist_single} close to the critical inverse temperature $\beta = 2 (1+ \varepsilon$), for a small $\varepsilon >0$. {For $\beta$ close to $2$, we use the approximation $r(\beta) \approx \sqrt{1 - \frac{2}{\beta}}$, which comes from the lower bound calculations for $r(\beta)$ in~\cite{bertini2010dynamical}[Sec. 2.1]}.}

Now, using the Taylor expansion for $\sqrt{1+\varepsilon}$, we have $\sqrt{\varepsilon(1+\varepsilon)} = \sqrt{\varepsilon} + O(\varepsilon \sqrt{\varepsilon})$. We set $\delta = \sqrt{\varepsilon}$ and, {ignoring terms of order ${\varepsilon}^2$ and above}, can then rewrite the equation for $H$ as:
\begin{align*}
    H f = f'' - \delta \cos(x) f - \delta^2 \sin^2(x) f + O(\delta^3) =: H_0 {f} + \delta H_1 {f}+ \delta^2 H_2 {f} + O(\delta^3).
\end{align*}
Therefore, we can consider $H$ as a small perturbation of the operator $H_0 f = f''$. 
Our goal is to calculate the eigenvalues of $H$ perturbatively for small $\delta$. We therefore consider the eigenvalue problem
\begin{equation}\label{eigenvalues}
    H \psi = E \psi.
\end{equation}
We expand $E$ and $\psi$ in power series in $\delta$:
\begin{align*}
    &\psi = \psi_0 + \delta \psi_1 + \delta^2 \psi_2 + \ldots \\
    &E = E_0 + \delta E_1 + \delta^2 E_2 + \ldots
\end{align*}
We substitute these expansions into~\eqref{eigenvalues} to obtain the following sequence of equations.
%\begin{align*}
%    (H_0 + \delta H_1 + \delta^2 H_2) (\psi_0 + \delta \psi_1 + \delta^2 \psi_2 + \ldots) = (E_0 + \delta E_1 + \delta^2 E_2 + \ldots) (\psi_0 + \delta \psi_1 + \delta^2 \psi_2 + \ldots)
%\end{align*}

\textbf{Order $O(1)$}: $H_0 \psi_0 = E_0 \psi_0$.

The eigenvalues of $H_0$ are given by $E_0^{(m)} = -m^2$ for $m \in \mathbb{N}$, and the corresponding eigenfunctions are $\psi_0^{(m)}(x) = A_m \cos(mx) + B_m \sin(mx)$, $A_m, B_m \in \R$. Due to the symmetry of the interaction potential, we will only consider even eigenfunctions, i.e. $B_m = 0$ for all $m \in \mathbb{N}$. We take $A_m = \frac{1}{\sqrt{\pi}}$ so that $\langle \psi_0, \psi_0 \rangle = 1$.
%This is just the eigenvalue equation for the unperturbed operator $H_0 f = f''$. Therefore, $E_0$ is just an unperturbed eigenvalue, and $\psi_0$ is its corresponding eigenfunction.

\textbf{Order $O(\delta)$}: $H_0 \psi_1 + H_1 \psi_0 = E_0 \psi_1 + E_1 \psi_0$.

We take the inner product of both sides with $\psi_0$ to obtain
\begin{align*}
    \langle \psi_0, H_0 \psi_1 \rangle + \langle \psi_0, H_1 \psi_0 \rangle = E_0 \langle \psi_0, \psi_1 \rangle + E_1 \langle \psi_0, \psi_0 \rangle.
\end{align*}
Since $H_0$ is self-adjoint in $L^2(0,2\pi)$, we obtain
\begin{align*}
    E_1 = \frac{\langle \psi_0, H_1 \psi_0 \rangle}{\langle \psi_0, \psi_0 \rangle} =0,
\end{align*}
{where the last equality comes from evaluating $\langle \psi_0, H_1 \psi_0 \rangle$} explicitly.
%due to symmetry.
%Furthermore:
%\begin{align*}
%    \langle \psi_0, H_1 \psi_0 \rangle %&= -\int_0^{2\pi} (A_n \cos(nx) + B_n %\sin(nx))^2 \cos(x) \dd x = 0
%\end{align*}
%Therefore, $E_1 = 0$. 

\textbf{Order $O(\delta^2)$}: $H_0 \psi_2 + H_1 \psi_1 + H_2 \psi_0 = E_0 \psi_2 + E_1 \psi_1 + E_2 \psi_0 = E_0 \psi_2 + E_2 \psi_0$.

We again take the inner product with $\psi_0$ and use the self-adjointness of $H_0$ to obtain
\begin{align*}
    E_2 = \frac{\langle \psi_0, H_1 \psi_1 + H_2 \psi_0 \rangle}{\langle \psi_0, \psi_0 \rangle} = \langle \psi_0, H_1 \psi_1 + H_2 \psi_0 \rangle.
\end{align*}
We now calculate $\psi_1$. It satisfies the equation
\begin{align*}
    H_0 \psi_1 + H_1 \psi_0 = E_0 \psi_1,
\end{align*}
which takes the form
\begin{align*}
    \psi_1'' + m^2 \psi_1 =  \frac{1}{\sqrt{\pi}} \cos(x) \cos(mx) ,
\end{align*}
in $(0, 2 \pi)$ with periodic boundary conditions.
This is a second order inhomogeneous ODE; its solution is
\begin{eqnarray*}
    \psi_1(x) &=& \frac{1}{(4m^2-1)\sqrt{\pi}} \Big( 2 m \sin(x) \sin(mx) + \cos(x) \cos(mx)\Big) + C_2 \sin(mx) + C_1 \cos(mx),
\end{eqnarray*}
for $C_1, C_2 \in \mathbb{R}$.
% We can now compute $\langle \psi_0, H_1 \psi_1 \rangle$. This has different values depending on the value of $n$.

% \begin{eqnarray*}
%     \langle \psi_0, H_1 \psi_1 \rangle = -\frac{1}{4n^2-1} && \int_0^{2\pi} \Big( A_n \cos(nx) + B_n \sin(nx) \Big) \cos(x) \Big( 2A_n n \sin(x) \sin(nx) + A_n \cos(x) \cos(nx) \\
%     &&-2B_n n \sin(x) \cos(nx) + B_n \cos(x) \sin(nx) \\
%     && + (4n^2 -1)(C_2 \sin(nx) + C_1 \cos(nx)) \Big) \dd x 
% \end{eqnarray*}

%We will need the following integrals:
%\begin{align*}
%    &\int_0^{2\pi} \cos(nx) \sin(x) \sin(nx) \cos(x)  \dd x = 
%    \begin{cases}
%        0 \text{ if } n \neq 1 \\
%        \frac{\pi}{4} \text{ if } n = 1
%    \end{cases}\\
%    &\int_0^{2\pi} (\cos(nx) \cos(x))^2 \dd x = 
%    \begin{cases}
%        \frac{\pi}{2} \text{ if } n \neq 1 \\
%        \frac{3\pi}{4} \text{ if } n = 1 
%    \end{cases}\\
%    &\int_0^{2\pi} (\cos(nx))^2 \cos(x) \sin(x) \dd x = 0 \\
%    &\int_0^{2\pi} \cos(nx) (\cos(x))^2 \sin(nx) \dd x = 0 \\
%    &\int_0^{2\pi} \cos(nx) \cos(x) \sin(nx) \dd x = 0 \\
%    &\int_0^{2\pi} (\cos(nx))^2 \cos(x) \dd x = 0 \\
%    &\int_0^{2\pi} (\sin(nx))^2 \cos(x) \sin(x) = 0 \\
%    &\int_0^{2\pi} \sin(nx) (\cos(x))^2 \cos(nx) \dd x = 0 \\
%    &\int_0^{2\pi} (\sin(nx) \cos(x))^2 \dd x = 
%    \begin{cases}
%        \frac{\pi}{2} \text{ if } n \neq 1 \\
%        \frac{3\pi}{4} \text{ if } n = 1 
%    \end{cases}\\
%    &\int_0^{2\pi} (\sin(nx))^2 \cos(x) \dd x = 0 \\
%    &\int_0^{2\pi} \sin(nx) \cos(x) \cos(nx) = 0
%\end{align*}
% After a tedious calculation, we conclude that that we have the following cases:

We can now calculate $\langle \psi_0, H_1 \psi_1 \rangle$:
\begin{align*}
        \langle \psi_0, H_1 \psi_1 \rangle = 
        \begin{cases}
         - \frac{5}{12} \text{ if } m = 1\\
        -\frac{1}{2(4m^2-1)} \text{ if } m \neq 1.
    \end{cases}
\end{align*}
%
% Case 1: $n \neq 1$. Then:
% \begin{align*}
%     \langle \psi_0, H_1 \psi_1 \rangle = -\frac{\pi}{2(4n^2-1)} ( A_n^2 + B_n^2 )
% \end{align*}

% Case 2: $n = 1$. Then:
% \begin{align*}
%     \langle \psi_0, H_1 \psi_1 \rangle = -\frac{\pi}{4 (4n^2-1)} (3(A_n^2 + B_n^2) + 2n(A_n^2 - B_n^2)) = - \frac{\pi}{12}(3(A_1^2 + B_1^2) + 2(A_1^2 - B_1^2)).
% \end{align*} 
%
Similarly, we have:
\begin{align*}
    &\langle \psi_0, H_2 \psi_0 \rangle =
    %-\int_0^{2\pi} (A_n \cos(nx) + B_n \sin(nx))^2 \sin^2(x) \dd x \\
    % &=-A_n^2 \int_0^{2\pi} (\cos(nx)\sin(x))^2 \dd x - 2A_n B_n \int_0^{2\pi} \sin(nx)\cos(nx)\sin^2(x) \dd x - B_n^2 \int_0^{2\pi} (\sin(nx)\sin(x))^2 \dd x \\
    \begin{cases}
        -\frac{1}{2} \text{ if } m \neq 1 \\
        -\frac{1}{4}  \text{ if } m = 1.
    \end{cases}    
\end{align*} 
We conclude that for $m = 1$ the second order perturbation is $E_2 = -\frac{2}{3}$, and for $m \neq 1$, $E_2 = - \frac{1}{2(4m^2-1)} - \frac{1}{2}$.

The eigenvalue of the perturbed operator $H$ is hence given by:
{
\begin{align*}
    E = E_0 + \delta E_1 + \delta^2 E_2 = -m^2 + \delta^2 E_2.
\end{align*}}
We are primarily interested in the first nonzero eigenvalue {as this represents the spectral gap of the operator}. {In Figure~\ref{fig:FirstEval}} we plot the asymptotic formula $E = -1 - \frac{2}{3} \delta^2$ versus the numerically obtained ones. The latter are obtained by using Matlab's \texttt{eig} function on the {(discretized)} operator $H f = f'' - \delta \cos(x) f - \delta^2 \sin^2(x) f $. {We note that on the scale of the figure the numerical and asymptotic eigenvalues are indistinguishable, even for moderately large $\delta$.}
%{\bf GP: maybe add some discussion about the numerical method. BG: I can do that if Benedetta can clarify which equation we find the eigenvalues of.}. As expected, for $\delta = \sqrt{\epsilon}$ sufficiently small, the agreement is very good. %Here is a plot that illustrates the behaviour of this expression as $\delta$ varies between $0.01$ and $0.25$. The plot also illustrates the eigenvalues of the perturbed operator found via Matlab to confirm our calculations.

\begin{figure}[H]
 \centering
    \includegraphics[width = 0.6\textwidth]{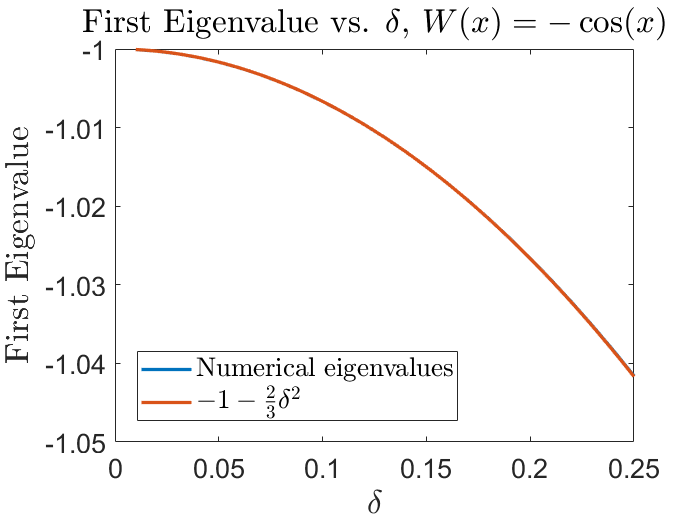}
    \caption{First eigenvalue as a function of $\delta$, $W(x) = -\cos(x)$.}
    \label{fig:FirstEval}
\end{figure}

\subsection{A higher order harmonic potential}
We now consider the interaction potential $W(x) = -\cos(nx)$. The Schr\"{o}dinger operator becomes: 
\begin{align*}
    H f = f'' - \frac{n^2}{2} \left(\frac{\beta^2}{2} r^2 \sin^2(nx)+ \beta r \cos(nx) \right)f.
\end{align*}
Writing $\beta = 2(1+\varepsilon)$ and using the same expansion as before, we obtain:
\begin{align*}
    Hf = f'' - \delta n^2 \cos(nx) f - \delta^2 n^2 \sin^2(nx) f + O(\delta^3) =: H_0 + \delta H_1 + \delta^2 H_2 + O(\delta^3).
\end{align*}

We note that we can indeed use the same expansion for $r$ as before, as this comes from using the Taylor expansion on $\exp( \beta r \cos(kx))$ and all terms involving the $\cos(kx)$ terms simplify. More precisely, we now have:
\begin{align*}
    r = \frac{\int_0^{2\pi} \cos(nx) \exp(\beta r \cos(nx)) \dd x}{\int_0^{2\pi} \exp(\beta r \cos(nx)) \dd x}.
\end{align*}
Approximating $e^x \approx 1 + x + \frac{x^2}{2}$ gives
\begin{align*}
    r &= \frac{\int_0^{2\pi} \cos(nx) + \beta r \cos^2(nx) + \frac{\beta^2}{2} r^2 \cos^3(nx) \dd x}{\int_0^{2\pi} 1 + \beta r \cos(nx) + \frac{\beta^2}{2} r^2 \cos^2(nx) \dd x} = \frac{\beta r \pi}{2\pi + \frac{\beta^2}{2} r^2 \pi}.
\end{align*} 
% \begin{align*}
%     1 = \frac{1 + \varepsilon}{1 + (1+\varepsilon)^2 r^2},
% \end{align*}
Substituting $\beta = 2(1 + \varepsilon)$, we deduce that
\begin{align*}
    r = \frac{\sqrt{\varepsilon}}{1+\varepsilon}.
\end{align*}

% \textcolor{blue}{Reviewer 2 suggested that the formula for $r$ above is missing a $\sqrt{2}$. I do not think this is the case as $r = \frac{\beta r \pi}{2\pi + \frac{\beta^2}{2} r^2 \pi}$ simplifies to:
% \begin{align*}
%     & 1 = \frac{\beta \pi}{2\pi + \frac{\beta^2 r^2 \pi}{2}} \\
%     & \beta = 2 + \frac{\beta^2 r^2}{2}
% \end{align*}
% substituting $\beta = 2(1+\varepsilon)$:
% \begin{align*}
%     & 2(1+\varepsilon) = 2 + \frac{4(1+\varepsilon)^2 r^2}{2} \\
%     & 1 + \varepsilon = 1 + (1+\varepsilon)^2 r^2 \\
%     & \varepsilon = (1+\varepsilon)^2 r^2 \\
%     & r = \frac{\sqrt{\varepsilon}}{1+\varepsilon}
% \end{align*}
% BEN: I agree that our calculation is correct (I checked it independently).
% }
% {\bf GP: I am not sure I understand this remark. Let's discuss}
% \begin{rem}
%     In the previous section, we used the approximation $r \approx \sqrt{\frac{\varepsilon}{1+\varepsilon}}$, following from the lower bound in~\cite[Eqn. 2.4]{bertini2010dynamical}. 
%     %and it is (at least for the $W(x) = -\cos(x)$ case) slightly more accurate. One can see by plotting the two functions on MATLAB but also by noticing that $\sqrt{\frac{\varepsilon}{1+\varepsilon}}$ was given as a lower bound for $r$, and $\frac{\sqrt{\varepsilon}}{1+\varepsilon} < \sqrt{\frac{\varepsilon}{1+\varepsilon}}$.
%     However, due to the fact that we then later on use the approximation $\sqrt{\varepsilon(1+\varepsilon)} = \sqrt{\varepsilon} + O(\varepsilon \sqrt{\varepsilon})$, the lower bound in Bertini's paper then leads to the same equations when analysing the perturbations.  {}
% \end{rem}
We study again the equations corresponding to different orders of $\delta = \sqrt{\varepsilon}$.

\textbf{Order $O(1)$}: $H_0 \psi_0 = E_0 \psi_0$.

As before, this is just the eigenvalue equation for the unperturbed operator $H_0 f = f''$. Therefore, $\psi_0(x) = \frac{1}{\sqrt{\pi}} \cos(mx)$ for $m \in \mathbb{N}$. 
%\textcolor{blue}{BEN: before we didn't have the minus sign in $\psi_0$; it obviously doesn't matter, but perhaps we should be consistent.}

\textbf{Order $O(\delta)$}: $H_0 \psi_1 + H_1 \psi_0 = E_0 \psi_1 + E_1 \psi_0$.

Using the same reasoning as with the previous case, we obtain
\begin{align*}
    E_1 = \langle \psi_0, H_1 \psi_0 \rangle
%\end{align*}
%Now 
%\begin{align*}
%    \langle \psi_0, H_1 \psi_0 \rangle = - n^2 A_m^2 \int_0^{2\pi} (\cos(mx))^2 \cos(nx) \dd x 
= 
    \begin{cases}
        0 & \text{ if } {m \neq \frac{n}{2}} \\
        -\frac{n^2}{2} & \text{ if } {m \neq \frac{n}{2}}.
    \end{cases}
\end{align*}
Therefore, if $n$ is even, we now have a non-zero first order perturbation; the corresponding eigenvalues of $H$ are
\begin{align*}
    E = -\frac{n^2}{2} \left(\frac{1}{2} + \delta \right).
\end{align*}

\textbf{Order $O(\delta^2)$}: $H_0 \psi_2 + H_1 \psi_1 + H_2 \psi_0 = E_0 \psi_2 + E_1 \psi_1 + E_2 \psi_0$.

As before,
\begin{align*}
    E_2 = \langle \psi_0, H_1 \psi_1 + H_2 \psi_0 \rangle.
\end{align*}

We only look at the second perturbation in the case where $E_1 = 0$, i.e. when $m \neq \frac{n}{2}$. Performing the same calculations as in the previous subsection, we obtain that for $n \neq m$:
\begin{align*}
    E_2 = \frac{n^2}{(n^2-4m^2)} \left(\frac{n^2}{2} \right) - \frac{n^2}{2},
\end{align*}
while for $n = m$:
\begin{align*}
    E_2 &= - \frac{5n^2}{12} - \frac{n^2}{4} = \frac{2n^2}{3}.
\end{align*}

Therefore, we obtain the following eigenvalues for $H$:
\begin{itemize}
    \item For $n \in \mathbb{N}$ even, $E_n = -\frac{n^2}{2} \left( \frac{1}{2} + \delta \right)$;
    \item For $n, m \in \mathbb{N}$, $m \neq \frac{n}{2}$, $E_n = -m^2 + \delta^2 \left( \frac{n^2}{(n^2-4m^2)} \left(\frac{n^2}{2} \right) - \frac{n^2}{2} \right)$;
    \item For $n \in \mathbb{N}$, $E_n = -n^2 - \delta^2 \left(\frac{n^2}{6} \right)$. 
    %\textbf{Is this missing a condition, otherwise it's a generalization of the first bullet point}
\end{itemize}

% {\bf GP: do we need the remark below?}
% Note: if we take $n =1$, then the eigenvalues for $m \neq 1$ turn out to be $E = -m^2 + \delta^2 \left( \frac{n^4}{2(n^2-4m^2)} - \frac{n^2}{2} \right) = -m^2 - \delta^2 \left( \frac{1}{2(4m^2-1)} + \frac{1}{2} \right)$ which agrees with our previous calculations.
{We have also computed these eigenvalues numerically for a few potentials with higher harmonics and for a fixed value of $\delta$, and verified that these agree with the formulae presented above. The following tables illustrate this agreement for the first few eigenvalues in the cases $W(x) = -\cos(nx)$, for $n \in \{2,3,4\}$.}
\begin{table}[h!]
    \caption{Eigenvalues for $W(x) = -\cos(2x)$, $\delta = 0.1$.} % title of Table
    \centering % used for centering table
    \begin{tabular}{c c c c} % centered columns (4 columns)
    \hline\hline %inserts double horizontal lines
    m & Perturbation eigenvalue & Numerical eigenvalue &  \\ [0.5ex] % inserts table
    %heading
    \hline % inserts single horizontal line
    1 & -1.2 & -1.2154  \\ % inserting body of the table
    2 & -4.0267 & -4.0267  \\
    3 & -9.0225 & -9.0221 \\
    4 & -16.0213 & -16.0213  \\ [1ex] % [1ex] adds vertical space
    \hline %inserts single line
    \end{tabular}
    \label{table:nonlin1} % is used to refer this table in the text
\end{table}

%{\bf GP: we need to explain the issue with the repeated eigenvalue.}

\begin{table}[h!]
    \caption{Eigenvalues for $W(x) = -\cos(3x)$, $\delta = 0.1$.} % title of Table
    \centering % used for centering table
    \begin{tabular}{c c c c} % centered columns (4 columns)
    \hline\hline %inserts double horizontal lines
    m & Perturbation eigenvalue & Numerical eigenvalue &  \\ [0.5ex] % inserts table
    %heading
    \hline % inserts single horizontal line
    1 & -0.964 & -0.9653  \\ % inserting body of the table
    2 & -4.1029 & -4.1015  \\
    3 & -9.06 & -9.0600  \\
    4 & -16.0524 & -16.0524  \\ [1ex] % [1ex] adds vertical space
    \hline %inserts single line
    \end{tabular}
    \label{table:nonlin2} % is used to refer this table in the text
\end{table}

\begin{table}[h!]
    \caption{Eigenvalues for $W(x) = -\cos(4x)$, $\delta = 0.1$.} % title of Table
    \centering % used for centering table
    \begin{tabular}{c c c c} % centered columns (4 columns)
    \hline\hline %inserts double horizontal lines
    m & Perturbation eigenvalue & Numerical eigenvalue &  \\ [0.5ex] % inserts table
    %heading
    \hline % inserts single horizontal line
    1 & -0.9733 & -0.9739  \\ % inserting body of the table
    2 & -4.8 & -4.8615  \\
    3 & -9.144 & -9.1434  \\
    4 & -16.067 & -16.067  \\ [1ex] % [1ex] adds vertical space
    \hline %inserts single line
    \end{tabular}
    \label{table:nonlin3} % is used to refer this table in the text
\end{table}

\subsection{Multichromatic potentials}
We now consider $W(x) = -\cos(x) - \frac{1}{2} \cos(2x)$. We recall that now the stationary distribution is given by
\begin{equation}\label{stationarydistcos2x}
    \rho_{\infty}(x) = \frac{\exp(\beta (r_1 \cos(x) + \frac{r_2}{2} \cos(2x)))}{\int_0^{2\pi}\exp(\beta (r_1 \cos(x) + \frac{r_2}{2} \cos(2x))) \dd x},
\end{equation}
with $r_1$ and $r_2$ satisfying the self-consistency equations
\begin{equation}\label{e:sc1}
    r_1 = \frac{\int_0^{2\pi} \cos(x)\exp(\beta (r_1 \cos(x) + \frac{r_2}{2} \cos(2x))) \dd x}{\int_0^{2\pi} \exp(\beta (r_1 \cos(x) + \frac{r_2}{2} \cos(2x))) \dd x}
\end{equation}
and
\begin{equation}\label{e:sc2}
    r_2 = \frac{\int_0^{2\pi} \cos(2x)\exp(\beta (r_1 \cos(x) + \frac{r_2}{2} \cos(2x))) \dd x}{\int_0^{2\pi} \exp(\beta (r_1 \cos(x) + \frac{r_2}{2} \cos(2x))) \dd x}.
\end{equation}
{We note that, for the non-trivial stationary state which corresponds to having exactly one non-zero coefficient $r_i$, the calculations become the same as the ones done in the previous subsection for a higher but single harmonic potential. Therefore, in this subsection we work with the case where both $r_1$ and $r_2$ are non-zero.}
We need to find approximations for $r_1$ and $r_2$ near the critical value of $\beta$. Using that $\exp(x) \sim 1 + x + \frac{x^2}{2}$ in equations (\ref{e:sc1})-(\ref{e:sc2}), we obtain that, close to $\beta = 2$, $r_2^2(\beta)$ behaves linearly, and that
\begin{equation}\label{e:r1r2approx}
    r_1 = \sqrt{r_2 \left( r_2 + \frac{2}{\beta} \right)}.
\end{equation}

To get an explicit expression for $r_2(\beta)$ we use a basic {regression} algorithm to obtain that the best approximation is given by
\begin{align*}
    r_2(\beta) = \sqrt{\frac{3}{2}-\frac{3}{\beta}}.
\end{align*}

Substituting this in (\ref{e:r1r2approx}), we obtain the corresponding approximating function for $r_1$:
\begin{align*}
    r_1(\beta) = \sqrt{\frac{3}{2} - \frac{3}{\beta} + \frac{2}{\beta}\sqrt{\frac{3}{2}-\frac{3}{\beta}}}
\end{align*}

We now use these approximations to study the Schr\"{o}dinger operator
for $W = -\cos(x)- \frac{1}{2} \cos(2x)$. With $\rho_{\infty}$ as in (\ref{stationarydistcos2x}), we have that
\begin{align*}
    U(x) := \frac{\beta}{2} W*\rho_{\infty}(x) = - \frac{\beta}{2}r_1 \cos(x) - \frac{\beta}{4} r_2 \cos(2x).
\end{align*}

Therefore
\begin{align*}
    &H f = f'' - \left( \frac{(U')^2}{4} + \frac{U''}{2} \right) f \\
    &= f'' - \left(\frac{\beta^2}{4} (r_1 \sin(x) + r_2 \sin(2x))^2 +  \frac{\beta}{2} (r_1 \cos(x) + 2r_2 \cos(2x))\right) f.
\end{align*}
We now substitute $\beta = 2(1 + \varepsilon)$, $\eta = {\varepsilon}^{1/4}$, use the above expressions for $r_1(\beta), r_2(\beta)$ and ignore terms of order higher than $O(\eta^3)$. After lengthy calculations, we obtain the following expression for the perturbed operator:
% \begin{align*}
%     \sqrt{\frac{3K^2}{2}-\frac{3K}{2}} = \sqrt{\frac{3}{2}} \delta
% \end{align*}

% We will however need to also introduce $\eta = \sqrt{\delta} = \varepsilon^{1/4}$.In particular we have:
% \begin{align*}
%     &\sqrt{\frac{3K^2}{2} - \frac{3K}{2} + \sqrt{\frac{3K^2}{2} - \frac{3K}{2}}} =\sqrt{\left(\sqrt{\frac{3}{2}\varepsilon(1+\varepsilon)} \right) \left(1 + \sqrt{\frac{3}{2}\varepsilon(1+\varepsilon)} \right)}
% \end{align*}
% Now write $\sqrt{1+x} = 1 + \frac{x}{2} + O(x^2)$. We use this with $x = \sqrt{\frac{3}{2} \varepsilon(1+\varepsilon)}$ to deduce:
% \begin{align*}
%     \sqrt{\frac{3K^2}{2} - \frac{3K}{2} + \sqrt{\frac{3K^2}{2} - \frac{3K}{2}}} &= \left( \frac{3}{2} \varepsilon \right)^{\frac{1}{4}} \sqrt{1+\frac{\varepsilon}{2}} \left(1 + \frac{1}{2} \sqrt{\frac{3}{2}\varepsilon + \frac{3}{2}\varepsilon^2} \right) \\
%     &= \left( \frac{3}{2} \right)^{1/4} \eta + O(\eta^3)
% \end{align*}
% Putting this together, for $K = 1 + \varepsilon$, $\eta = \varepsilon^{1/4}$, ignoring everything of order $O(\eta^3)$:
\begin{align*}
    Hf & := f'' - \left( \sqrt{\frac{3}{2}} \eta^2 \sin^2(x) + \left( \frac{3}{2} \right)^{1/4} \eta \cos(x) + 2 \sqrt{\frac{3}{2}} \eta^2 \cos(2x) \right) f \\
    &=: H_0 + \eta H_1 + \eta^2 H_2,
\end{align*}
where
\begin{align*}
    &H_0 f = f'' \\
    &H_1 f = - \left(\frac{3}{2} \right)^{1/4} \cos(x) f(x) \\
    &H_2 f = - \sqrt{\frac{3}{2}} \left( \sin^2(x) + 2 \cos(2x) \right) f(x) = - \frac{1}{2} \sqrt{\frac{3}{2}} (1 + 3\cos(2x)) f(x) .
\end{align*}
We follow the usual perturbation argument to find the eigenvalues $E$ and eigenfunctions $\psi$ solving $H \psi = E \psi$. We solve the equation order by order. 

\textbf{Order $O(1)$}: $H_0 \psi_0 = E_0 \psi_0$.

As usual, this gives us $\psi_0 = \frac{1}{\sqrt{\pi}} \cos(nx)$, $E_0 = - n^2$. 

\textbf{Order $O(\eta)$}: $H_0 \psi_1 + H_1 \psi_0 = E_0 \psi_1 + E_1 \psi_0$.

Performing the same steps as before gives us $E_1 = 0$ again.
% We take the inner product of both sides with $\psi_0$ to obtain:
% \begin{align*}
%     \langle \psi_0, H_0 \psi_1 \rangle + \langle \psi_0, H_1 \psi_0 \rangle = E_0 \langle \psi_0, \psi_1 \rangle + E_1 \langle \psi_0, \psi_0 \rangle
% \end{align*}
% Since $H_0$ is Hermitian and $\langle \psi_0, \psi_0 \rangle = 1$, we obtain:
% \begin{align*}
%     E_1 = \langle \psi_0, H_1 \psi_0 \rangle
% \end{align*}
% Now,
% \begin{align*}
%     \langle \psi_0, H_1 \psi_0 \rangle &= - \left( \frac{3}{2} \right)^{1/4} \frac{1}{\pi} \int_0^{2\pi} \cos(nx)^2 \cos(x) \dd x = 0
% \end{align*}

% Therefore, $E_1 = 0$. 

\textbf{Order $O(\eta^2)$}: $H_0 \psi_2 + H_1 \psi_1 + H_2 \psi_0 = E_0 \psi_2 + E_1 \psi_1 + E_2 \psi_0 = E_0 \psi_2 + E_2 \psi_0$.

We again take the inner product with $\psi_0$ and use that $H_0$ is Hermitian to obtain
\begin{align*}
    E_2 = \langle \psi_0, H_1 \psi_1 + H_2 \psi_0 \rangle.
\end{align*}

Solving similar equations as in the previous subsections we have, for $n = 1$:
\begin{align*}
    E_2 = - \frac{(36)^{1/4}}{12} - \frac{3}{2} \sqrt{\frac{3}{2}}
\end{align*}
and for $n \neq 1$:
\begin{align*}
    E_2 = -\frac{1}{2} \sqrt{\frac{3}{2}} \left( \frac{1}{4n^2-1} + 1 \right).
\end{align*}

Therefore, the eigenvalues of $H$ are
{
\begin{align*}
    E = 
    \begin{cases}
        -1 -\eta^2 \left( \frac{(36)^{1/4}}{12} + \frac{3}{2} \sqrt{\frac{3}{2}} \right) \text{ for } n = 1 \\
        -n^2 - \eta^2 \frac{1}{2} \sqrt{\frac{3}{2}} \left( \frac{1}{4n^2-1} + 1 \right) \text{ for } n \neq 1
    \end{cases}.
\end{align*}
}

Fixing $\eta$ (for example, $\eta = 0.5$) and calculating a few values, we can see that they are in very good agreement with the eigenvalues found numerically by MATLAB.
As before, {in Figure~\ref{fig:FirstEvalMulti},} we provide a plot to illustrate the behaviour of the first eigenvalue as $\eta$ varies.

\begin{figure}[H]
 \centering
    \includegraphics[width = 0.6\textwidth]{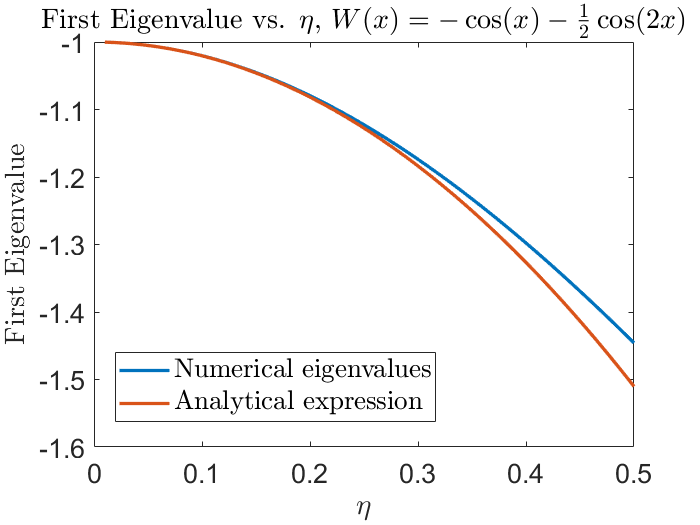}
    \caption{First eigenvalue as a function of $\eta$, $W(x) = -\cos(x)-\frac{1}{2} \cos(2x)$.}
    \label{fig:FirstEvalMulti}
\end{figure}

\section{Numerical experiments}
\label{sec:num_exp}

% \textbf{Ben suggestion:} I think we can combine the PDE and SDE dynamics here.  We don't need to do the SDE for all cases, but we may as well use it to validate the examples we want to show, which will reduce the number of figures we need.

In this section we consider both steady states and the dynamics of the McKean-Vlasov equation (\ref{e:mckean_external}) by solving it numerically using pseudospectral methods ~\cite{boyd2001chebyshev,nold2017pseudospectral,DDFTCode}.  
Furthermore, we compare the solution of the PDE to results obtained from Monte-Carlo simulations {of the $N$ particle system \eqref{e:IPS} using the Euler-Maruyama method~\cite{kloedenplaten2013}.}

\subsection{Steady states and intermediate dynamics}
\label{subsec: steady states}
As described in Section~\ref{sec:SelfConsistency}, we expect there to be a qualitative change in the nature of the steady-state solution of the PDE (\ref{e:mckean_external}) at $\beta_c$. Due to the gradient structure of the PDE, (numerical approximations to the) steady states can be determined either by solving the PDE directly over a long time interval, or by solving the self-consistency equation (\ref{e:lane_emden}) iteratively, for example via Picard iteration.  A key difference here is that iterative approaches can converge to unstable steady states, whereas the PDE method always approaches a stable steady state.  This motivates our choice in this section to use the long time PDE solutions.

In order to focus on the effects of the choice of the interaction parameters on the dynamics, in this section we consider dynamics with no external potential.  We will reintroduce the external potential in the following section when comparing against stochastic dynamics. {Unless otherwise specified in the captions, all figures have $V(x) = 0$.}

We reiterate from above that for for $\beta < \beta_c$, we expect the uniform state to be stable, corresponding to the long time solution of the PDE being $\rho_{\infty} = 1/(2\pi)$.  In contrast, for $\beta > \beta_c$, we expect to observe the other, peaked solutions. 
Furthermore, as $\beta$ grows larger, we expect the steady state to be more strongly peaked.

In Figures~\ref{fig:longtime1}--\ref{fig:longtime4} we show `intermediate' and `long' time dynamics of the PDE for various interactions, $W$, and a range of $\beta$.  The particular choices of $W$ are somewhat arbitrary; we have chosen interactions with $n=2$ and $n=3$ coefficients, with the values chosen to result in a range of different dynamics.

We start with an initial condition that is a small perturbation of the uniform state - $\rho_0(x) = \frac{1}{2\pi} + 0.01 \sin \left( x - \frac{\pi}{2} \right)$.  
We note that the particular form of the initial condition dictates the position of the peak in the steady state; we have chosen it such that the peak is located at $\theta = \pi$ for clarity. Due to the translational invariance of the problem, this perturbation can, in principle, be chosen to control the position of the peak(s); we will make use of this in the following section when comparing to the SDE dynamics. {While we use a one-peaked perturbation here, a multi-peaked initial condition does not significantly affect the dynamics, as long as the perturbation away from the uniform state is small.} One may think that, in parameter regimes where the uniform state is unstable, a deliberate perturbation is not necessary to see dynamics, as this should be produced by the accumulation of numerical errors.  However, due to the very high accuracy of the chosen numerical schemes we find that it is necessary to deliberately perturb the uniform state in order to induce dynamics in reasonable computational times.

As can be seen in the figures, for $\beta < \beta_c$, the long time solution is indeed the {uniform} steady state.  As expected, for $\beta>\beta_c$ we see a more interesting range of long time solutions, which (for the examples below) have a single peak.  Note that for $\beta$ close to $\beta_c$ the convergence to equilibrium can be very slow; this is demonstrated in the right hand plots of Figures~\ref{fig:longtime3} and~\ref{fig:longtime4} where the solutions for $\beta$ close to $\beta_c$ have not yet converged to the final steady state.  This effect is due to the exponential slowing down of the dynamics near the critical value of $\beta$.

Considering now the `intermediate' dynamics in the left hand plots of Figures~\ref{fig:longtime3} and~\ref{fig:longtime4}, we note that there are transient regimes where the number of peaks is the value of $k$ corresponding to the largest $|a_k|$. However, these states appear to be unstable as the long time dynamics results in a single peak. The instability of multi-peak steady states was also observed in \cite{primi2009mass}, \cite{geigant2012stability}.

Looking at the dynamics in a bit more detail, we see that for the interactions in Figures~\ref{fig:longtime1} and~\ref{fig:longtime2}, the largest (in magnitude) coefficient is $a_1$, which corresponds to a single Fourier mode.  As expected, this leads to relatively simple dynamics in which the solution quickly converges to a single peak.  As shown in Figure~\ref{fig:longtime2}, adding a third non-zero coefficient does not change the critical value $\beta_c$, nor the qualitative nature of the solution. However, it does affect the quantitative dynamics, for example by decreasing the width of the final solutions for fixed $\beta$.  There is also a more obvious three-peaked state during the dynamics for $\beta=2$, although it is possible that this arises at a different time for the dynamics in Figure~\ref{fig:longtime1}. 

In Figures~\ref{fig:longtime3} and~\ref{fig:longtime4} we consider interactions where the largest magnitude coefficient corresponds to a higher Fourier mode.  In these cases, we clearly see the transient states which have a number of peaks corresponding to the Fourier mode with the coefficient with the largest magnitude. However, the steady state solution is dominated by the lowest non-zero Fourier mode, which in the cases shown here results in a single peak. {We note that the unstable states are still observable in the numerical simulations due to the high accuracy of the computational methods used.}

%\textbf{To do:}
%\begin{itemize}
%\item Add times to all plots.  Also, Figure 6 is different to the others in that it shows two different potentials, but only one time; this may confuse the reader.  Could we make this into two figures and add an intermediate time?
%\item Give the explicit initial condition
%\item Change x axis label to $x$ (or $\theta$ - we should at least be consistent).
%\item Unify the titles, e.g., to include time, $V=0$, etc
%\item Rewrite the captions (I can do this once we've finalised the figures)
%\item Can we give the explicit $\beta_c$ for each of the cases?
%\item \textbf{Ben suggestion:} I think we should show some Picard steady state solutions here.  I also think it would be worth showing examples where we have multiple steady states and also demonstrating their stability as a function of $\beta$.
%We should probably start with some simple, mono-chromatic potentials, and then move on to the more complicated ones.
%\end{itemize}

\begin{figure}[H]
  \centering
  \begin{minipage}[b]{0.49\textwidth}
    \includegraphics[width=\textwidth]{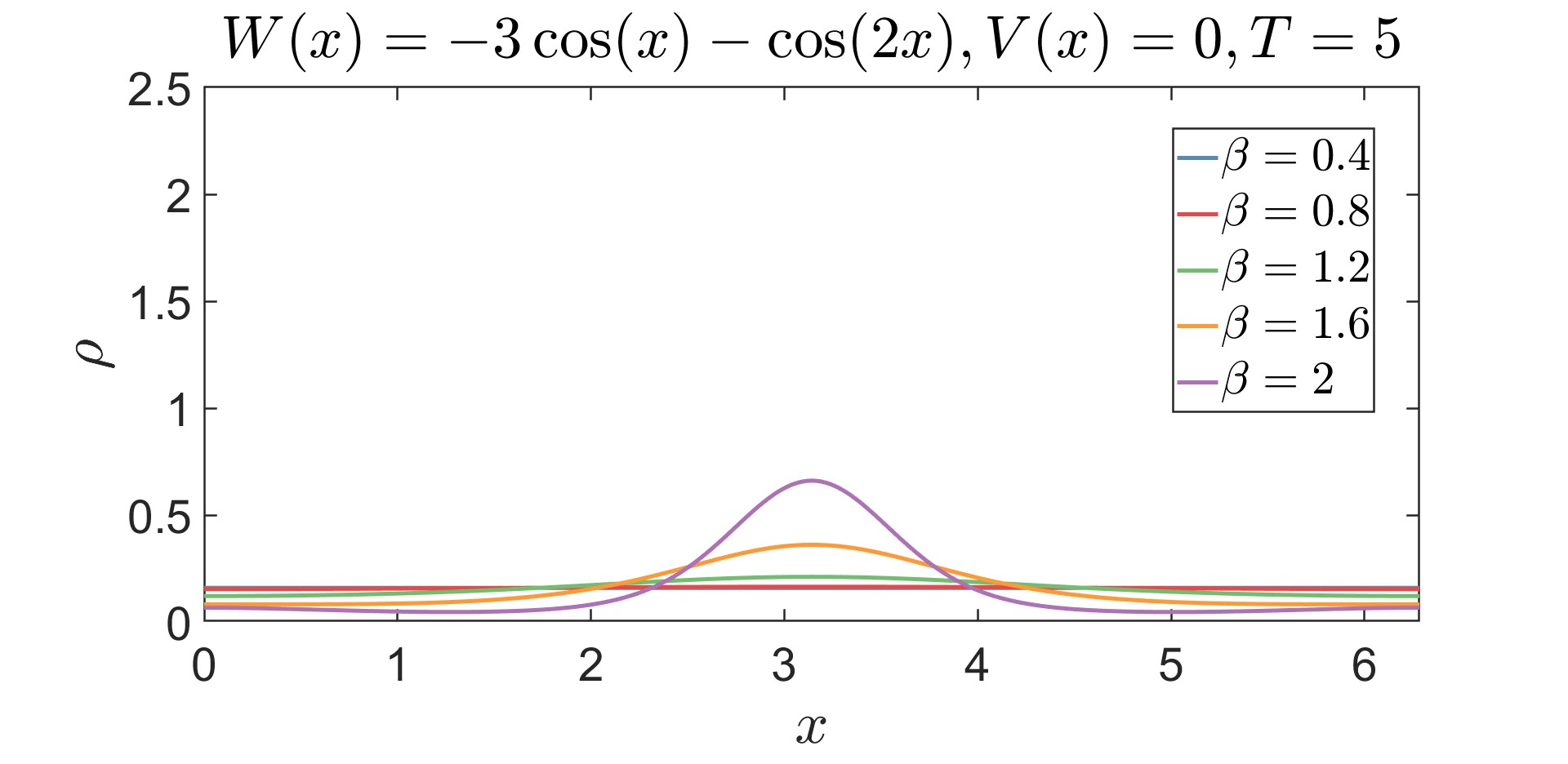}
  \end{minipage}
  \hfill
  \begin{minipage}[b]{0.49\textwidth}
    \includegraphics[width=\textwidth]{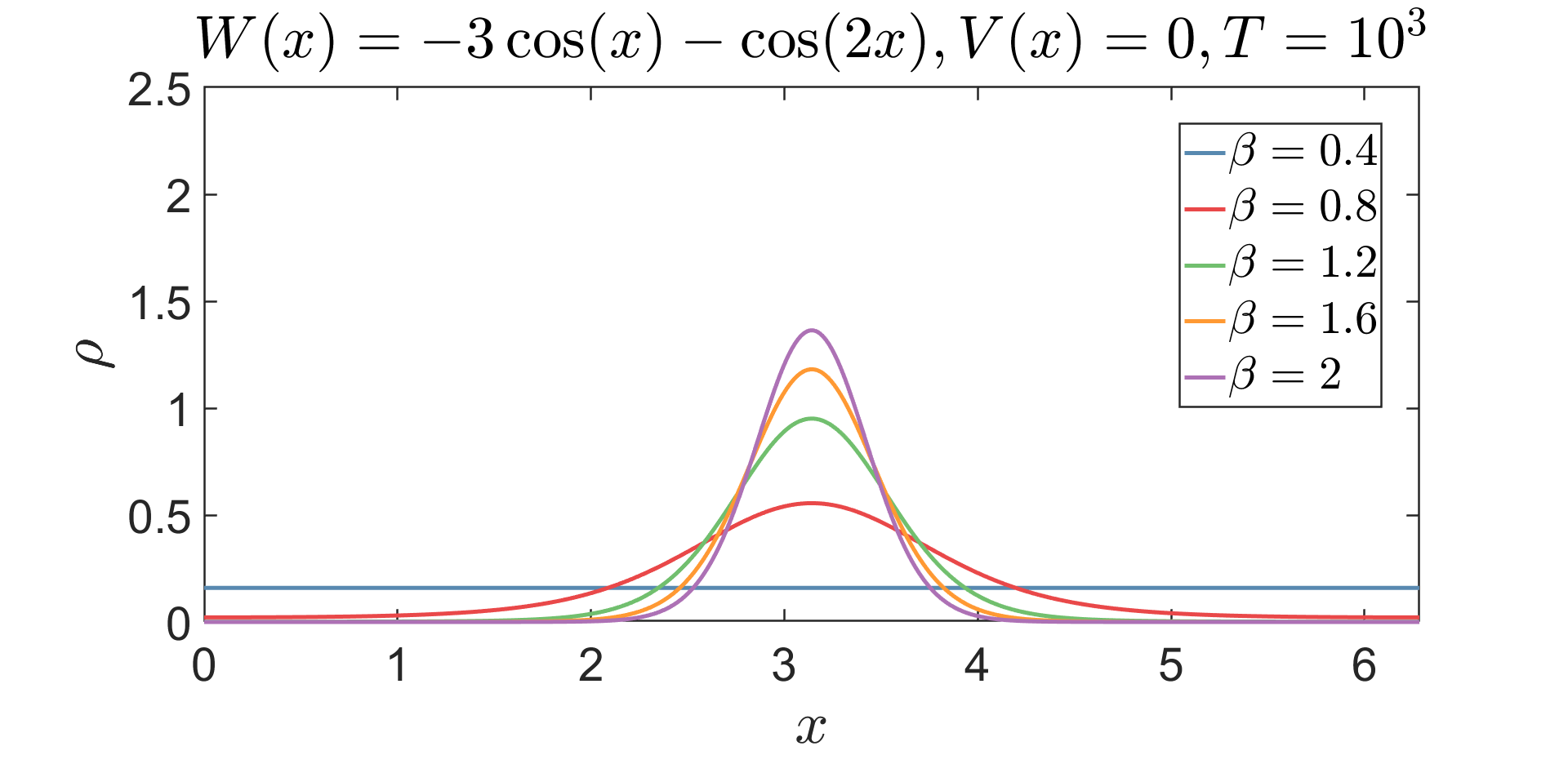}
  \end{minipage}
    \caption{PDE dynamics for different values of $\beta$ for the interacting potential {$W(x) = - 3\cos(x) - \cos(2x)$} for times $T = 5$ (left) and $T = 10^3$ (right). The critical temperature is $\beta_c = \frac{2}{3}$.}
    \label{fig:longtime1}
\end{figure}

\begin{figure}[H]
  \centering
  \begin{minipage}[b]{0.49\textwidth}
    \includegraphics[width=\textwidth]{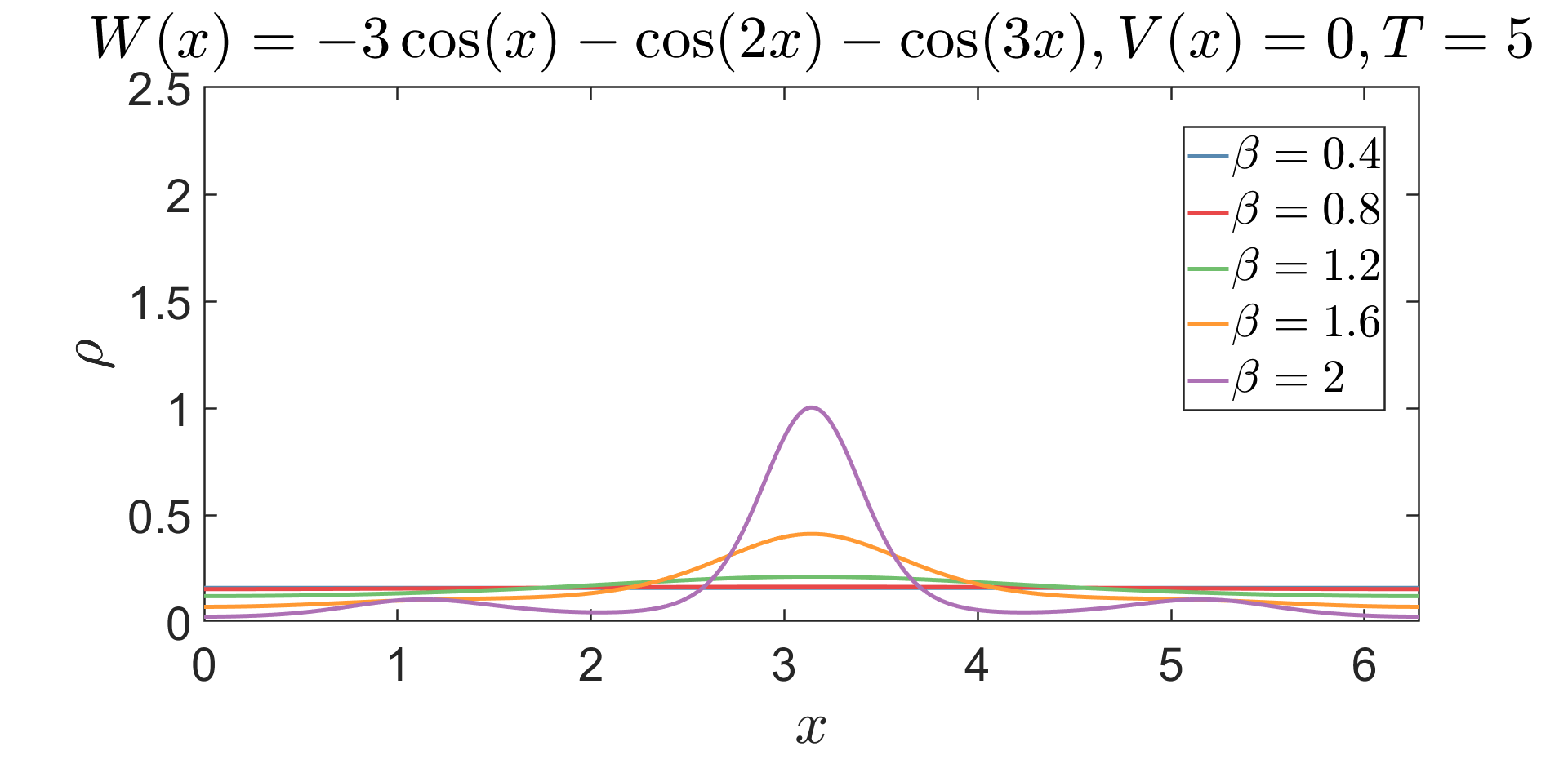}
  \end{minipage}
  \hfill
  \begin{minipage}[b]{0.49\textwidth}
    \includegraphics[width=\textwidth]{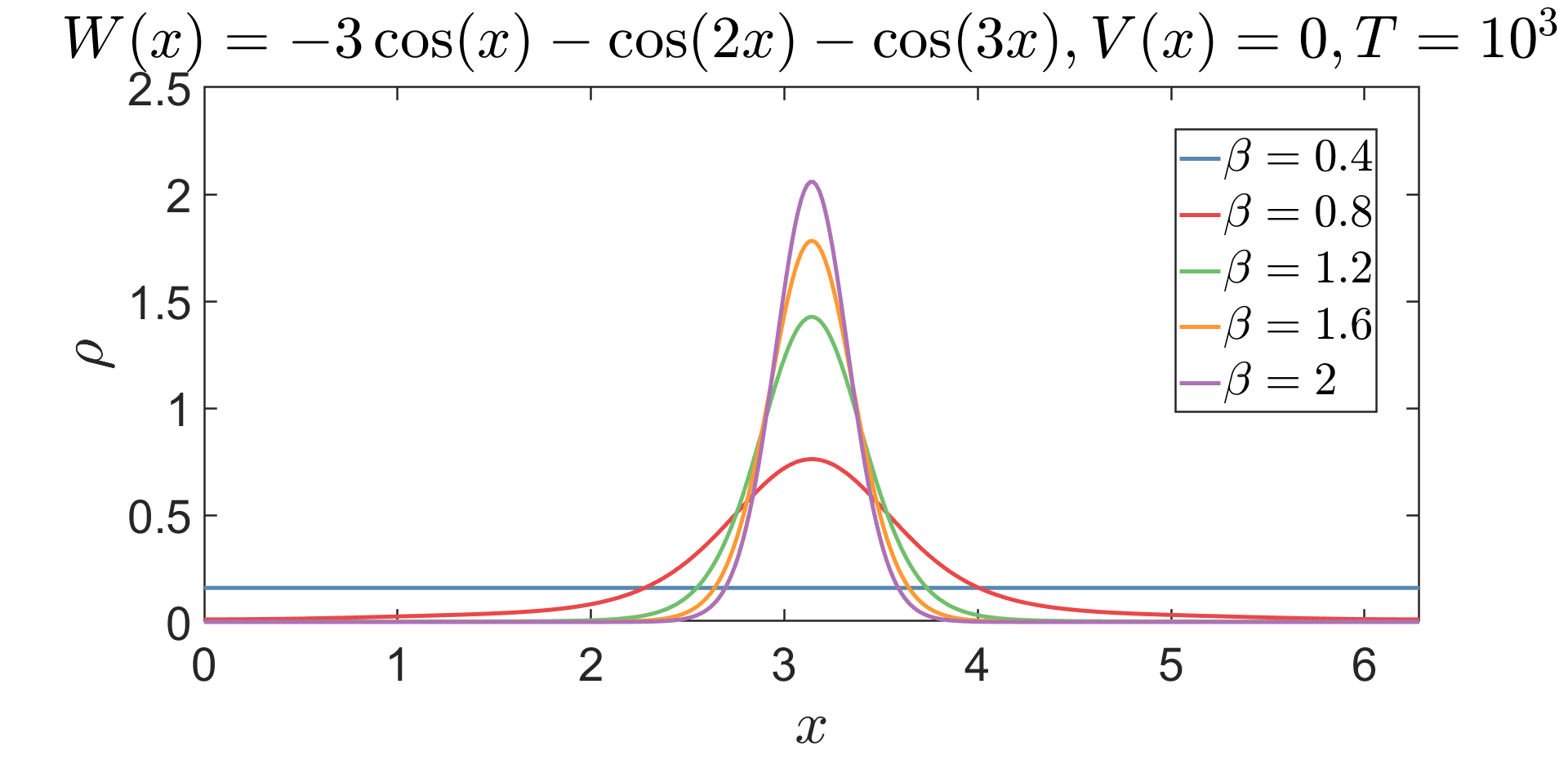}
  \end{minipage}
    \caption{PDE dynamics for different values of $\beta$ for the interacting potential {$W(x) = - 3\cos(x) - \cos(2x)-\cos(3x)$} for times $T = 5$ (left) and $T = 10^3$ (right). The critical temperature is $\beta_c = \frac{2}{3}$.}
    \label{fig:longtime2}
\end{figure}

\begin{figure}[H]
   \centering
    \begin{minipage}[b]{0.49\textwidth}
     \includegraphics[width=\textwidth]{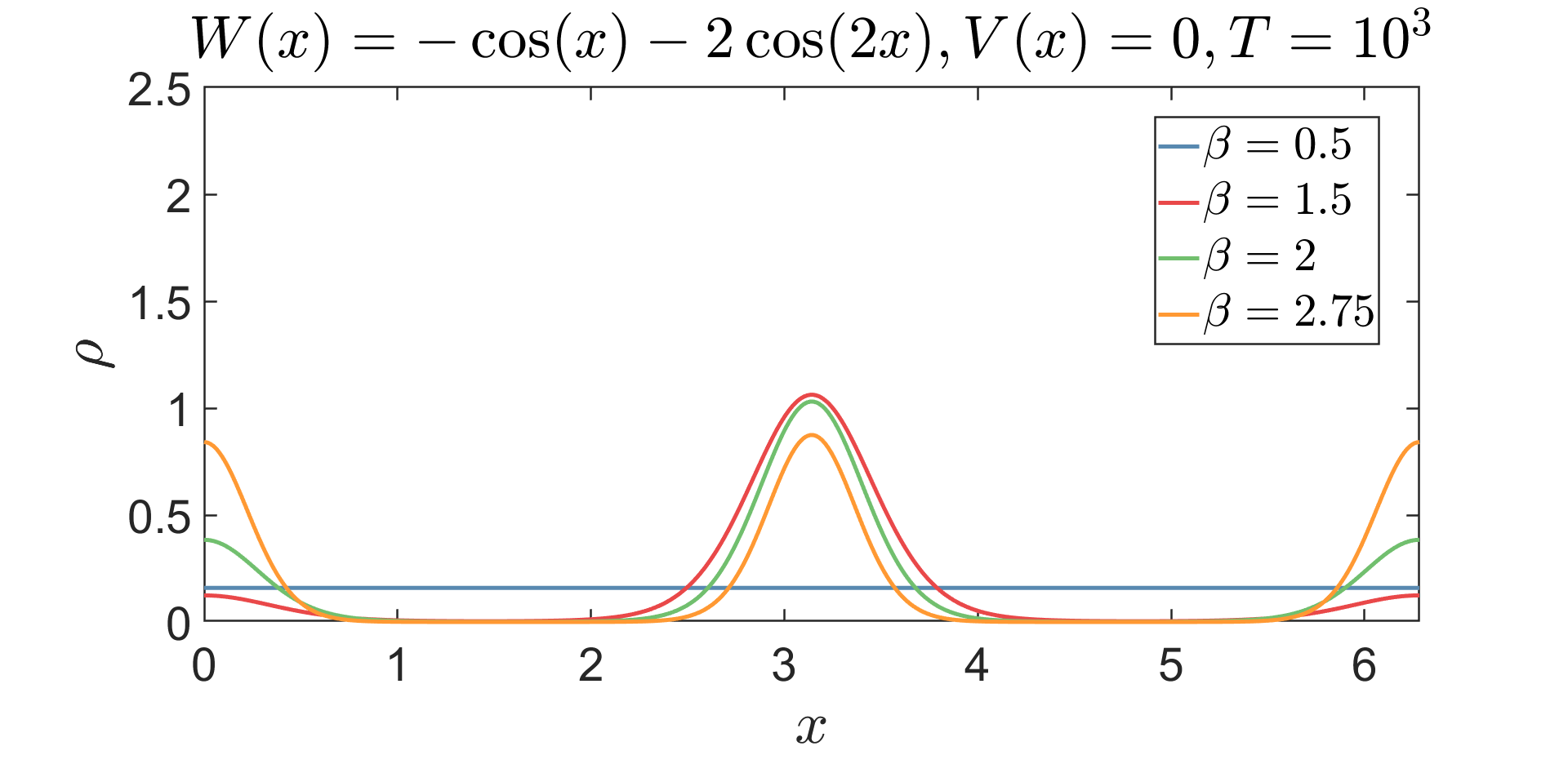}
   \end{minipage}
   \hfill
   \begin{minipage}[b]{0.49\textwidth}
     \includegraphics[width=\textwidth]{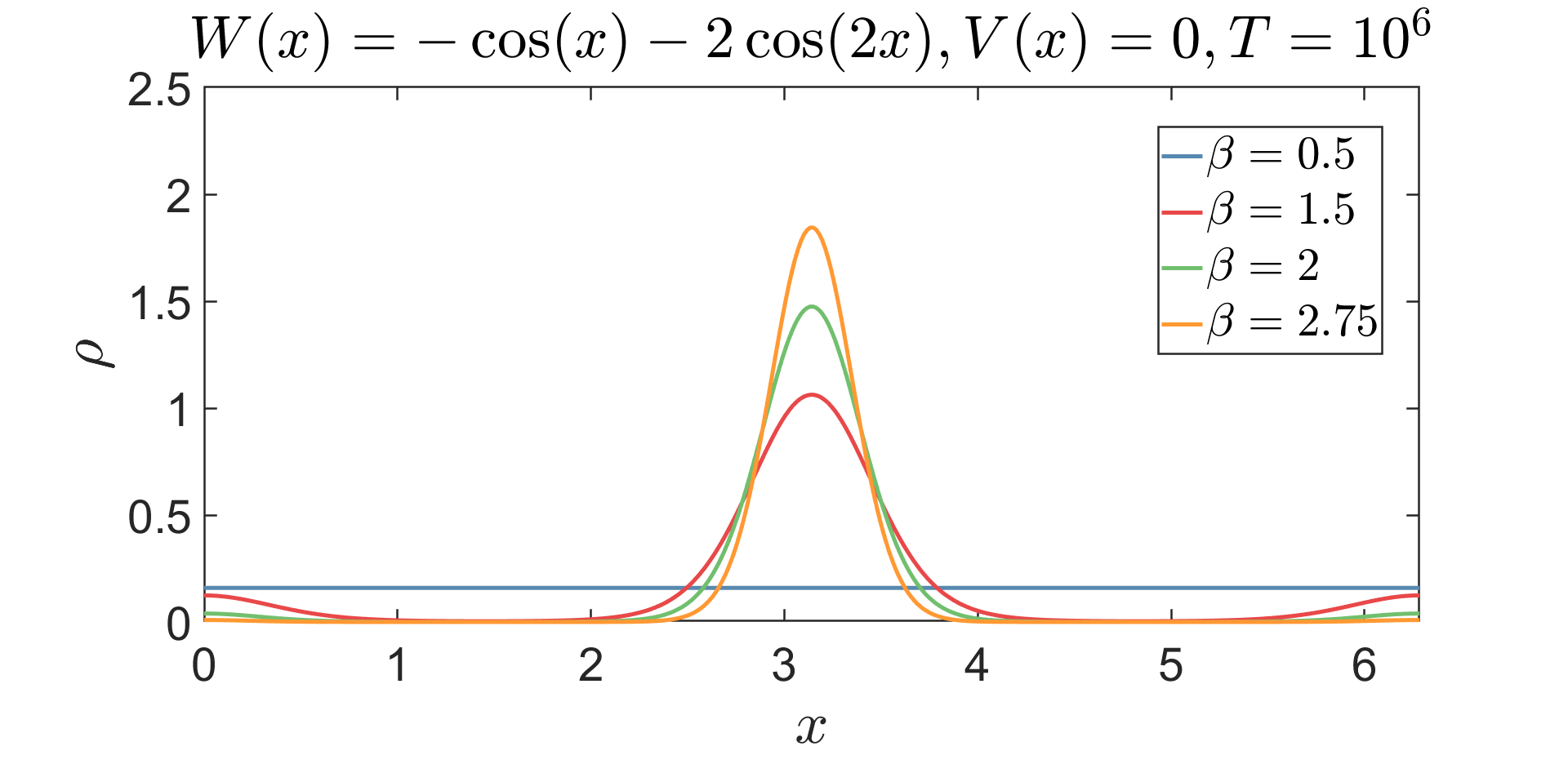}
   \end{minipage}
      \caption{PDE dynamics for different values of $\beta$ for the interacting potential $W(x) = -\cos(x)-2\cos(2x)$, for times $T = 10^3$ (left) and $T = 10^6$ (right). The critical temperature is $\beta_c = 1$.}
      \label{fig:longtime3}
\end{figure}

\begin{figure}[H]
  \centering
  \begin{minipage}[b]{0.49\textwidth}
    \includegraphics[width=\textwidth]{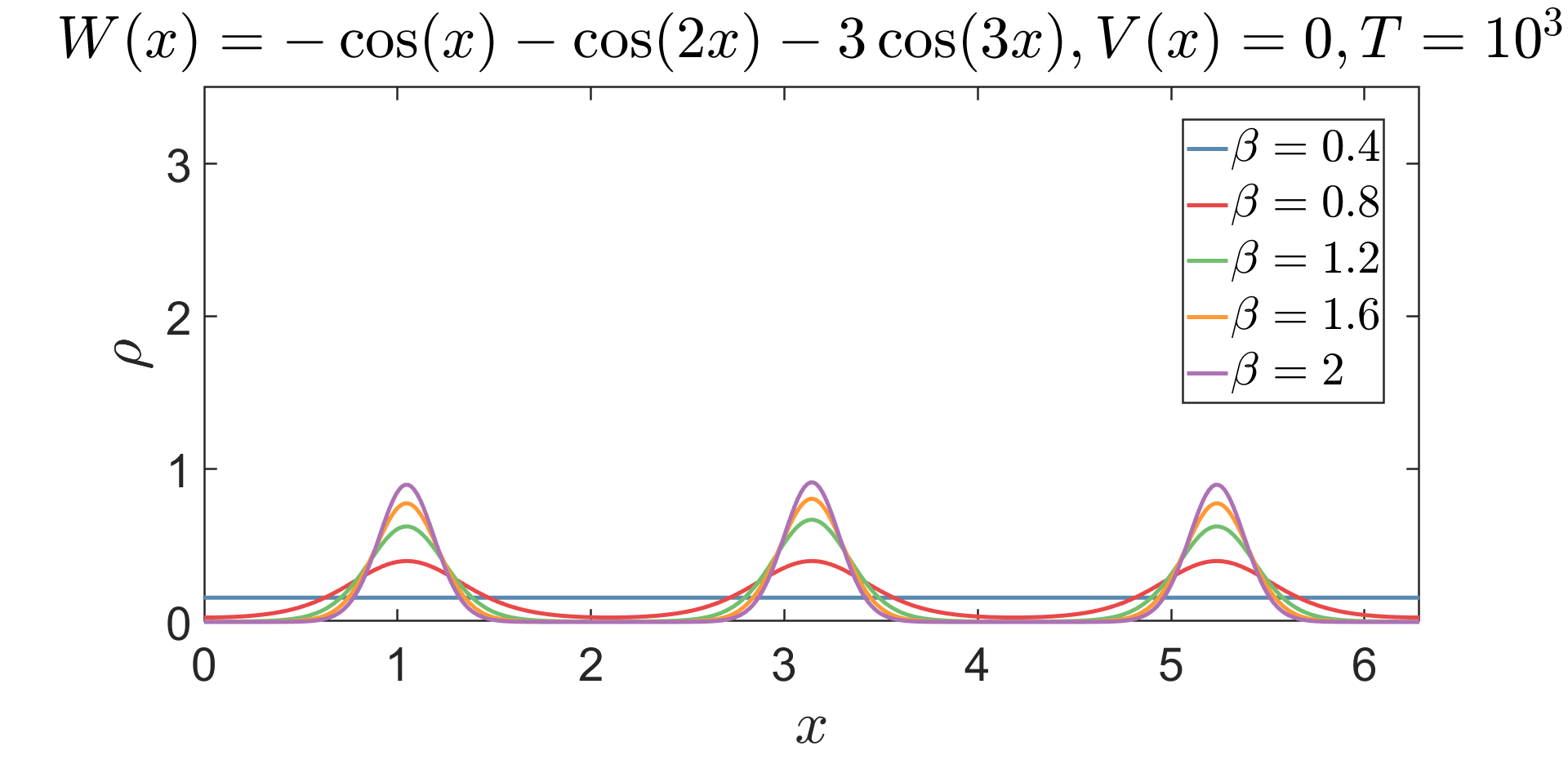}
  \end{minipage}
  \hfill
  \begin{minipage}[b]{0.49\textwidth}
    \includegraphics[width=\textwidth]{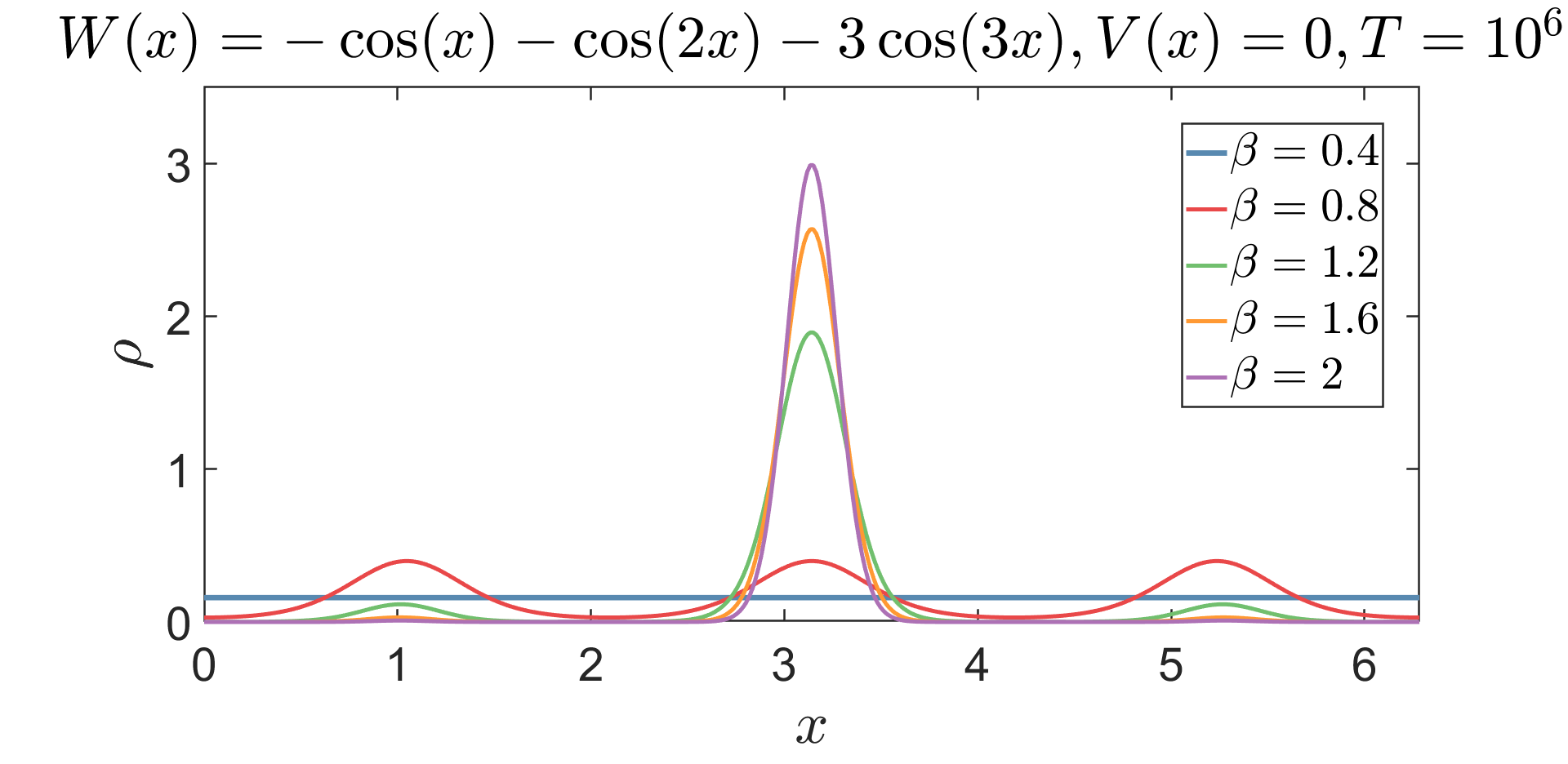}
  \end{minipage}
  \hfill
     \caption{PDE dynamics for different values of $\beta$ for the interacting potential $W(x) = -\cos(x)-\cos(2x)-3\cos(3x)$, for times $T = 10^3$ (left) and $T=10^6$ (right). The critical temperature is $\beta_c = \frac{2}{3}$.}
     \label{fig:longtime4}
\end{figure}

Figures \ref{fig:longtime1}-\ref{fig:longtime4} only show stable steady states with one single peak. It is possible to construct models with stable multipeak stationary distributions by simply rescaling the domain of the interaction potential, i.e. by considering potentials with higher harmonics and with a zero first Fourier mode. For example, the interaction potential $W(x) = - \cos(2x)$ will have a two-peak stable steady state.

\subsection{Comparison with Monte Carlo Simulations}

In this section we compare the dynamics of the PDE, \eqref{e:mckean_external}, with those of the underlying interacting particle SDE,~\eqref{e:IPS}.  The aim here is twofold: we will (i) demonstrate the effects of translational invariance on the results of stochastic sampling, and (ii) compare the PDE and SDE dynamics directly for a single initial condition with a large number of particles; this is a numerical demonstration of the mean-field limiting dynamics agreeing.  

We will also investigate systems which have a non-zero confining potential.  For a non-trivial potential, this breaks the translational invariance and allows us to demonstrate very good agreement between the two solutions by performing multiple runs of the stochastic dynamics and averaging.  This is computationally much cheaper than performing a single run with a larger number of particles: For $N$ particles and $R$ runs, the computational cost scales approximately as $N^2R$ (where the $N^2$ scaling results from having to compute the interaction potential), so increasing the number of runs is much cheaper than increasing the number of particles.

To numerically solve the SDE ~\eqref{e:IPS}, we use the Euler-Maruyama method~\cite{kloedenplaten2013} and compare our findings with the results obtained from the PDE solver. Unless otherwise stated, in all the simulations we use 500 particles, and a timestep of $\Delta t = 0.01$.  
%\textbf{Can we check this - one of them seems to have used this and others $0.01$; it would be good to be consistent, if possible.}  
The initial condition is generated by sampling from the same initial condition $\rho_0$ as for the PDE given in Section \ref{subsec: steady states} using Monte-Carlo slice sampling~\cite{neal2003slice}.

We first consider the interaction potential $W(x) = -\cos(x)-\frac{1}{2}\cos(2x)$, with no confining potential, an interaction strength $\kappa = 1$, and inverse temperature $\beta = 3$.  We run the dynamics up to a final time $T=200$.

As mentioned, in the absence of a confining potential, the problem is translationally invariant. This means, in particular, that above the phase transition we have infinitely many stationary states parameterized by an angle $\theta \in[0, 2\pi]$ \cite{bertini2010dynamical}. Consequently, when performing particle simulations, if we average over many realizations of the noise, we obtain (approximately) the uniform distribution; see the left plot in Figure~\ref{fig:cos_cos2x}. On the other hand, if we perform a single simulation, with a sufficiently large number of particles, then the results of the stochastic simulations are in good agreement with the results of the PDE, up to a translation in space; see the right plot in Figure~\ref{fig:cos_cos2x}.  Note that, in order to more clearly demonstrate the agreement, we have adjusted the PDE initial condition through a translational shift so that the positions of the peaks (approximately) align.  This gives exactly the same result as simply re-plotting the PDE (or SDE) solution on a translated axis.

\begin{figure}[H]
  \centering
  \begin{minipage}[b]{0.49\textwidth}
    \includegraphics[width=\textwidth]{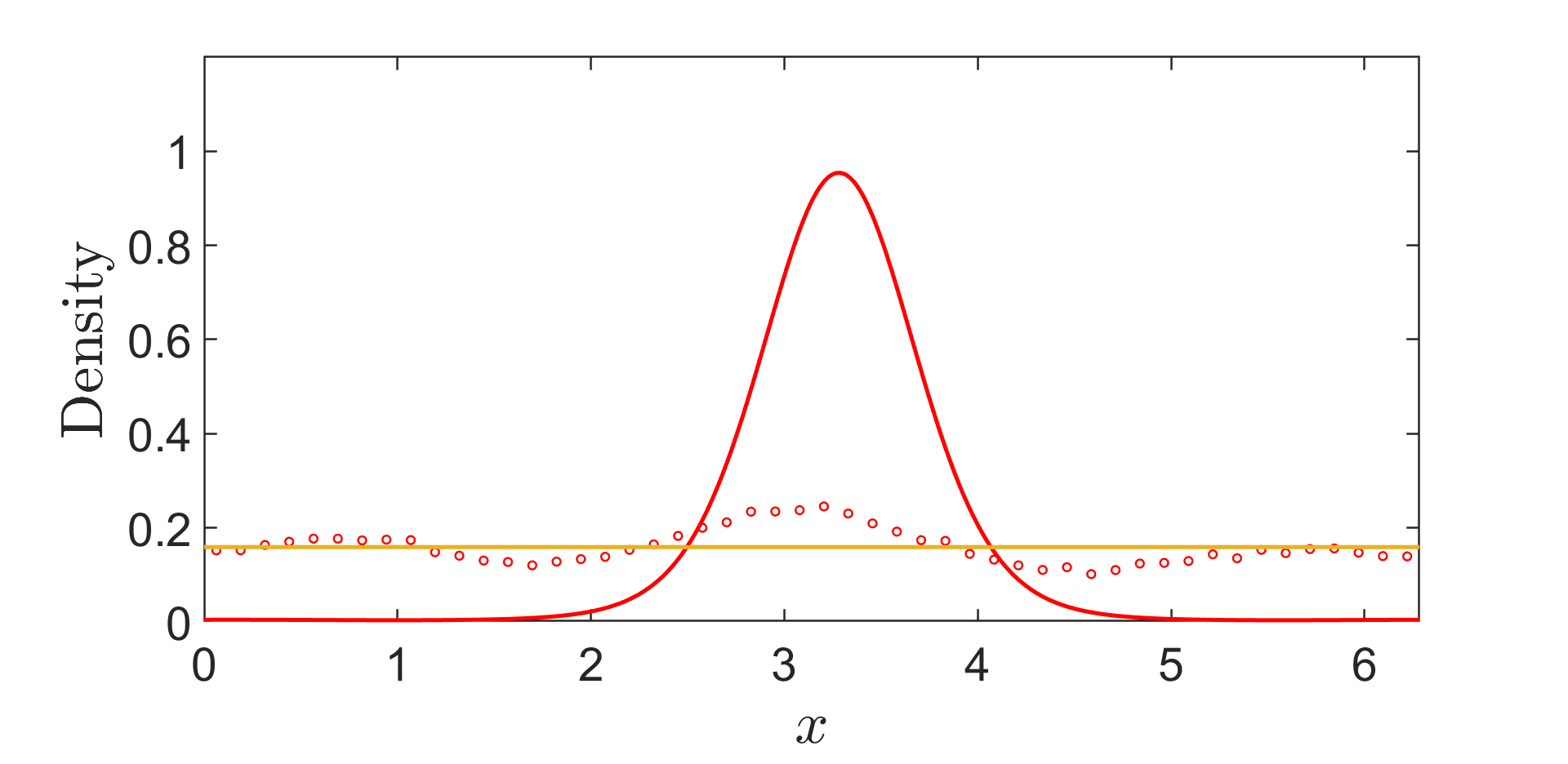}
  \end{minipage}
  \hfill
  \begin{minipage}[b]{0.49\textwidth}
    \includegraphics[width=\textwidth]{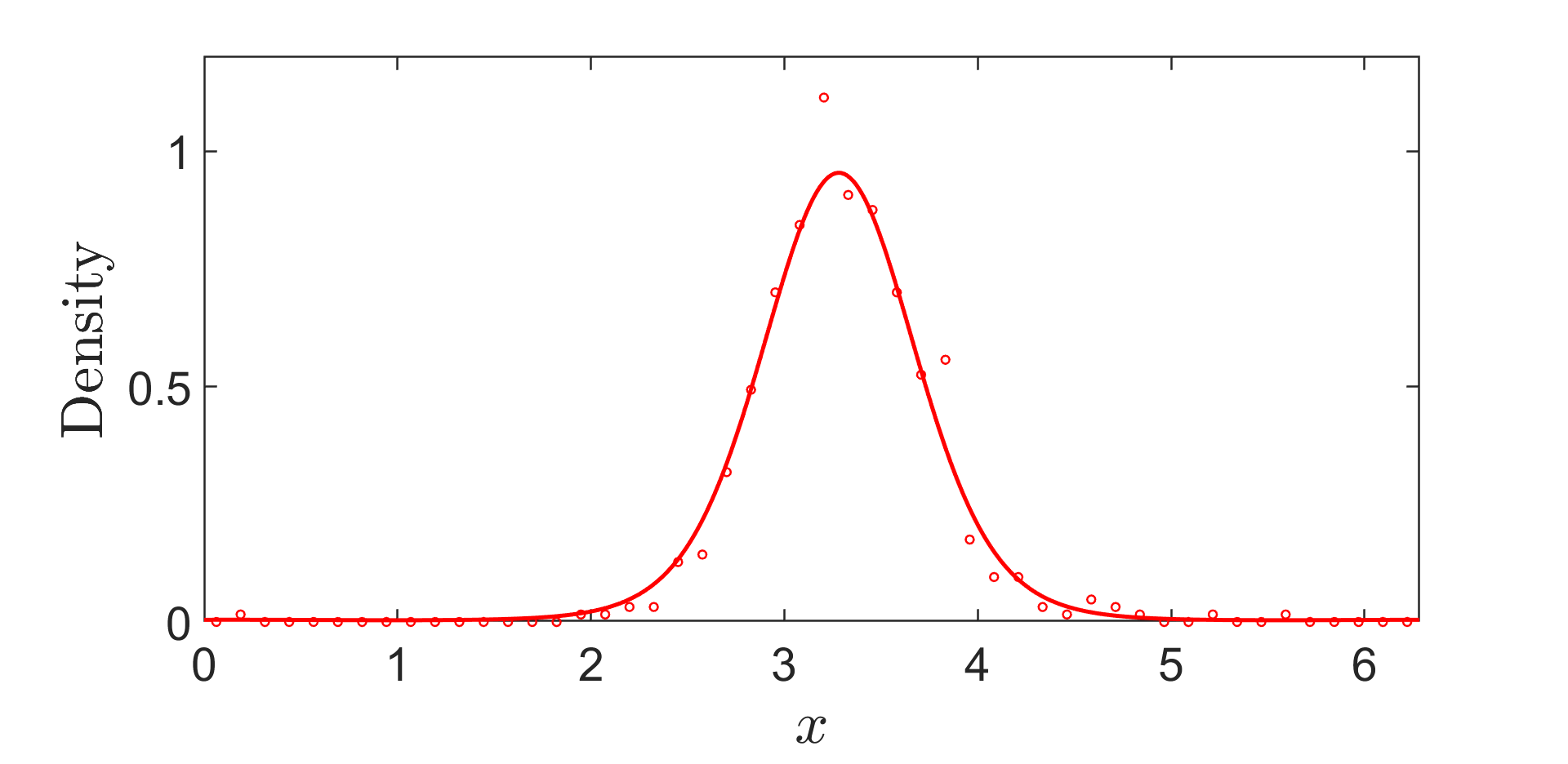}
  \end{minipage}
\caption{PDE and SDE dynamics for $W(x) = -\cos(x)-\frac{1}{2}\cos(2x)$ for 500 particles and 100 runs (left) and 1 run (right).}
%\textbf{Can we make these two plots more similar in terms of aspect ratio, line thicknesses, etc?  Ideally they will look identical apart from the SDE results}}
\label{fig:cos_cos2x}
\end{figure}

We now compare the PDE and SDE dynamics for a range of other interactions and no external potential.  Our next example concerns the interaction potential $W(x) = -\frac{1}{4} \cos(4x) - \frac{1}{6} \cos(6x)$, and is chosen to demonstrate the agreement for multi-peaked solutions.  %\textbf{What's the end time and time step here?}. 
We integrate up to $T = 200$. %and the nSteps value is tMax*100
 As can be seen in Figure~\ref{fig:PDESDE4peaks}, the agreement between the PDE and SDE is very good. Note that here we did not need to shift the PDE solution to align the peaks; this would not be true for a different realisation of the noise.

\begin{figure}[H]
    \centering
    \includegraphics[width = 0.75\textwidth]{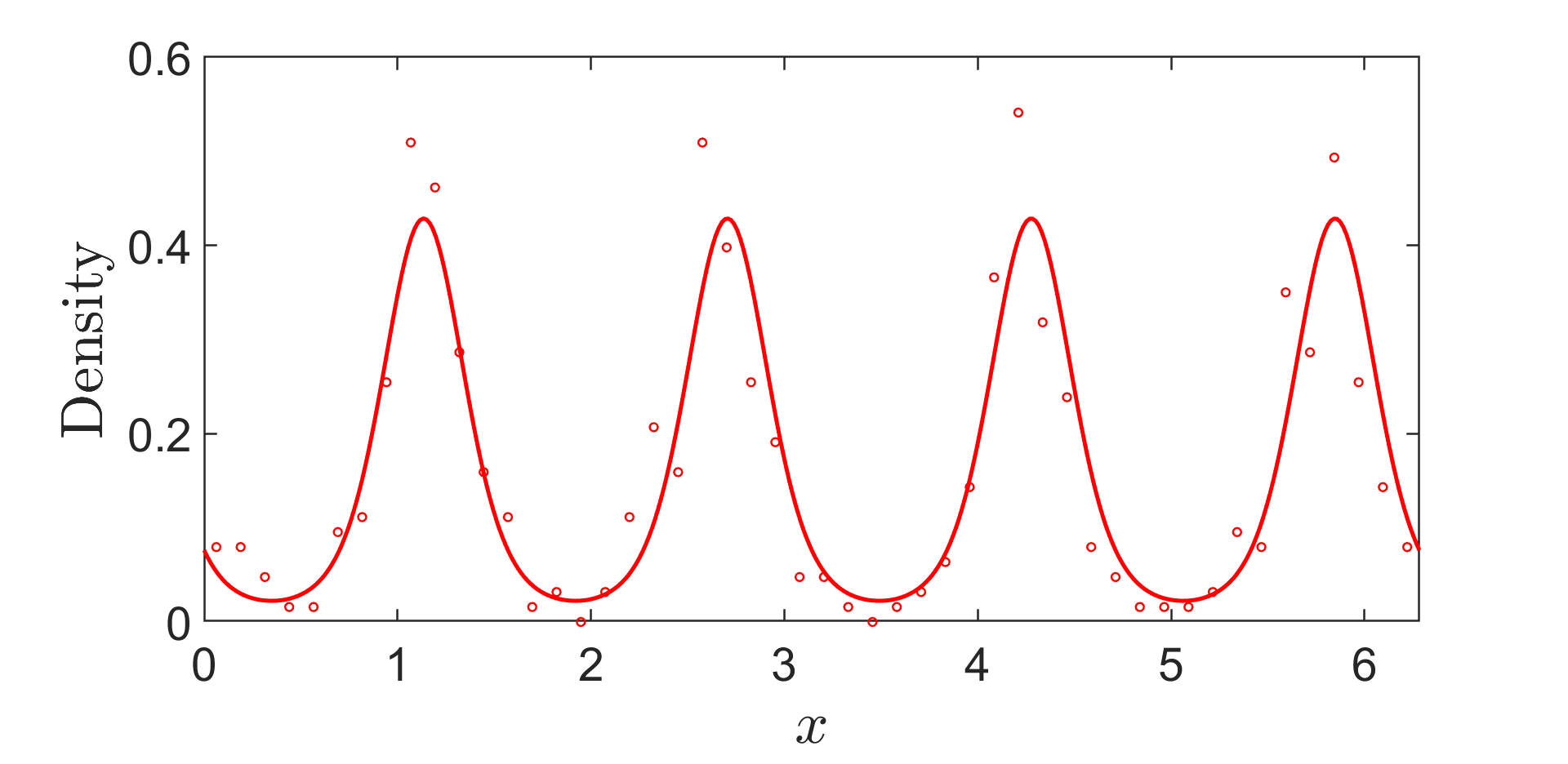}
    \caption{PDE and SDE dynamics with one run for $W(x) = -\frac{1}{4}\cos(4x) - \frac{1}{6}\cos(6x), \beta = 10$.}
    \label{fig:PDESDE4peaks}
\end{figure}

Our next example is chosen from~\cite{vukadinovic2023phase} concerning Hodgkin-Huxley oscillators, where the second Fourier mode now has positive sign. We take $W(x) = -\cos(x) + \frac{1}{2}\cos(2x)$. Our numerical experiments (not shown) agree with the result of the paper that there is a phase transition at $\beta_c = 2$. We plot here the density for $\beta = 3$, with {N = $100$}, T = $200$. %\textbf{Need the timestep}.  
The results are shown in Figure~\ref{fig:PDESDEHH}; again, the agreement between the SDE and PDE is very good. 
{We note that the Hodgkin-Huxley model includes a positive Fourier mode, and so is not included in the space of interaction potentials considered in our main framework; however, our numerical simulations are equally applicable to such potentials, allowing us to also analyse this example.}
%\textbf{Did we do any aligning of this?} Yes.

\begin{figure}[H]
    \centering
    \includegraphics[width = 0.75\textwidth]{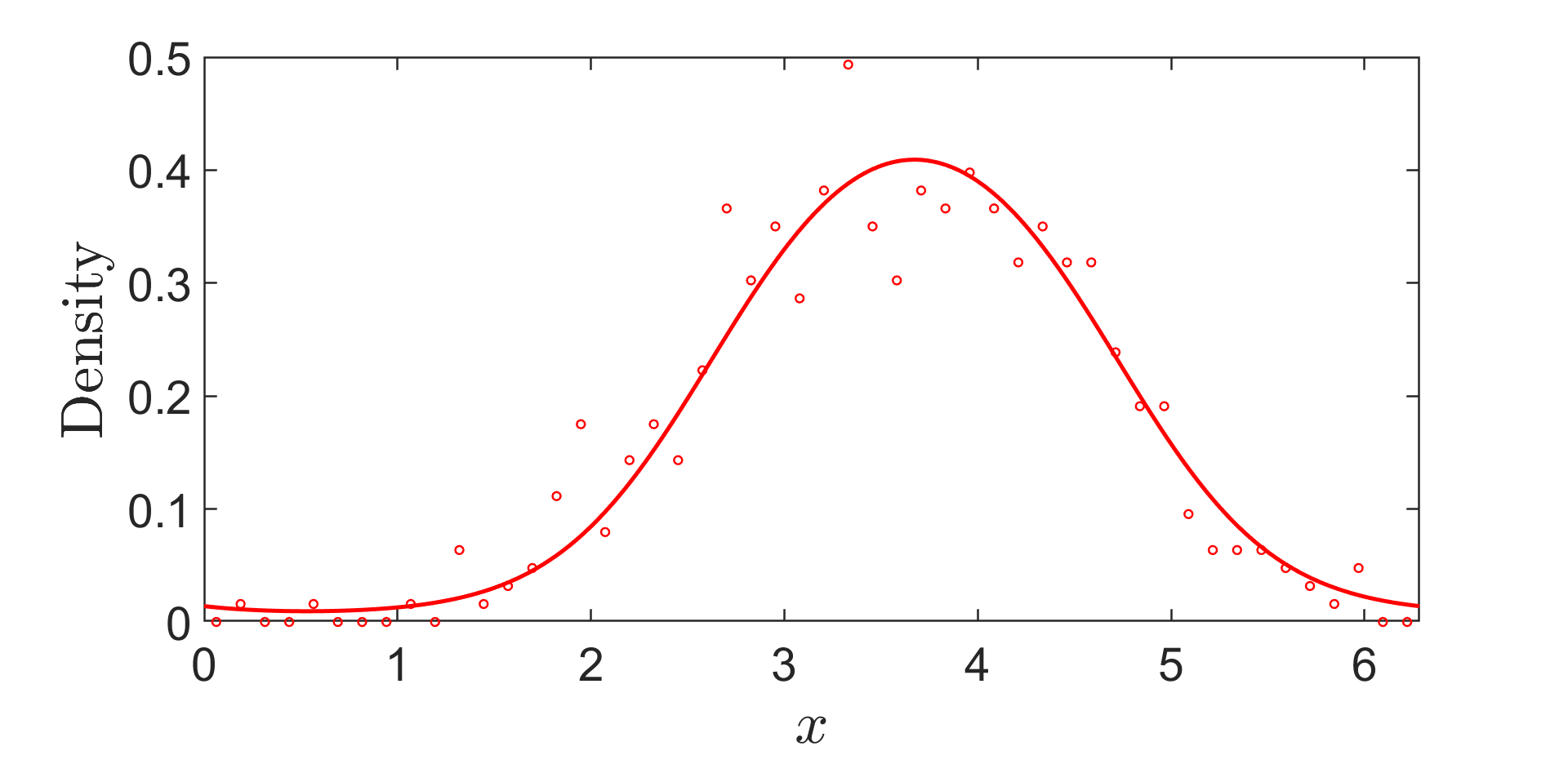}
    \caption{PDE and SDE dynamics for $W(x) = -\cos(x) + \frac{1}{2}\cos(2x), V(x) = 0, \beta = 3$.}
    \label{fig:PDESDEHH}
\end{figure}

%{\bf for Greg: add the calculation of stationary states for the $-cos(x)$ confining and interaction potentials, discuss about the stability of the steady states.}

In contrast to the examples presented so far, the presence of an external/confining potential breaks translation invariance. Consequently, averaging over many realization of the noise in the particle simulations leads to a non-uniform stationary state. We now consider two such examples, averaging solutions for 500 particles over 10 runs in each case.

For the first example we consider the Brownian mean field model in a magnetic field, \eqref{e:BMF_magnetic}, with $W(x) = -\cos(x)$ and $V(x) = 0.2 \cos(x+\pi)$. %\textbf{Need final time and timestep.} 
We integrate up to $T = 1000$.
The results are shown in Figure~\ref{fig:PDESDEVcos}, where we once again have excellent agreement between the PDE and SDE.
\begin{figure}[H]
    \centering
    \includegraphics[width = 0.75\textwidth]{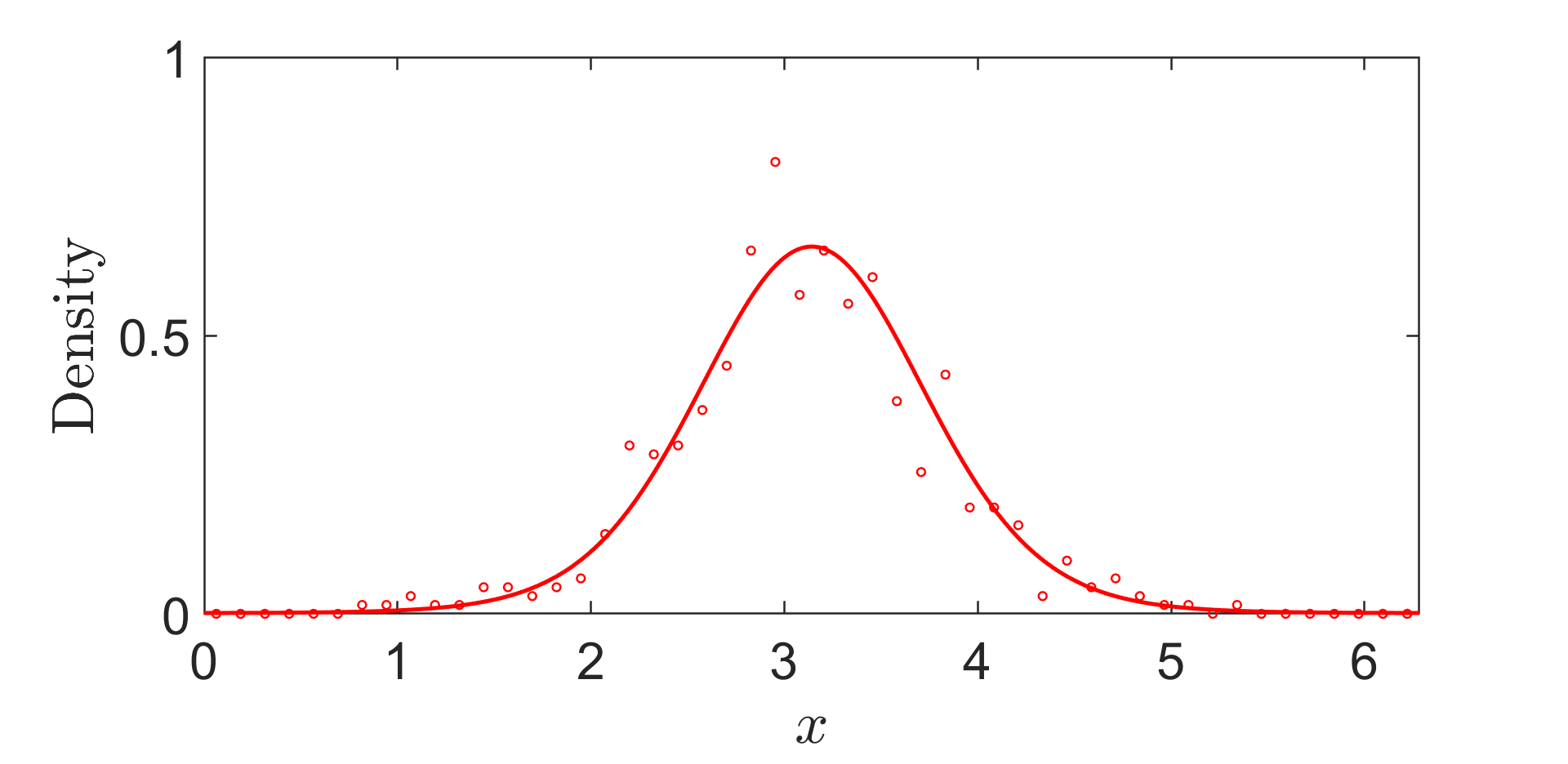}
    \caption{PDE and SDE dynamics for $W(x) = -\cos(x)$ with confining potential $V(x) = 0.2 \cos(x+\pi), \beta = 3$, 500 particles, 10 runs.}
    \label{fig:PDESDEVcos}
\end{figure}

Finally for long time solutions, we consider an example taken from \cite{angeli2023well}. Here $W(x) = -2\cos(2x)$ and $V(x) = \cos(x)$. In this case the dynamics converge more slowly and we take a final time of $T = 7000$. Figure~\ref{fig:PDESDEMichela} demonstrates the excellent agreement between the dynamics at long times.
\begin{figure}[H]
    \centering
    \includegraphics[width=0.75\textwidth]{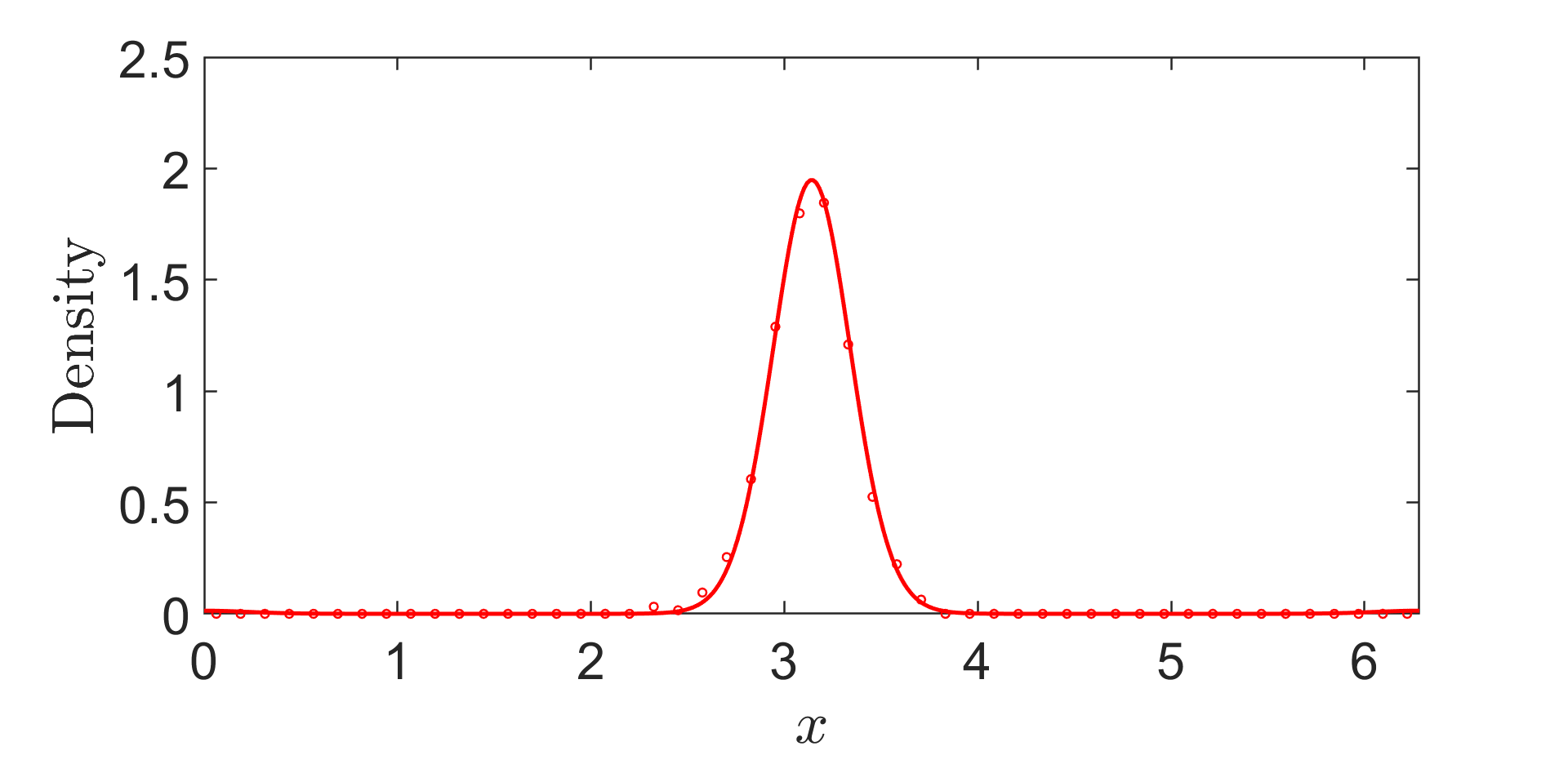}
    \caption{PDE and SDE dynamics for $W(x) = -2\cos(x)$ with confining potenital $V(x) = \cos(x+\pi), \beta = 3, 500$ particles.}
    \label{fig:PDESDEMichela}
\end{figure}

Of course, one is often interested not only in long time dynamics, but also in the full evolution of the density.  In Figure \ref{fig:Snapshots} we compare the PDE and SDE dynamics at a number of times for two different combinations of interaction and external potentials. As previously discussed, the agreement between the two dynamics is better in the right hand plot as we have a confining potential which breaks translation invariance, and a simpler interaction potential.

\begin{figure}[H]
    \centering 
    \begin{minipage}[b]{0.49\textwidth}
        \includegraphics[width=\textwidth]{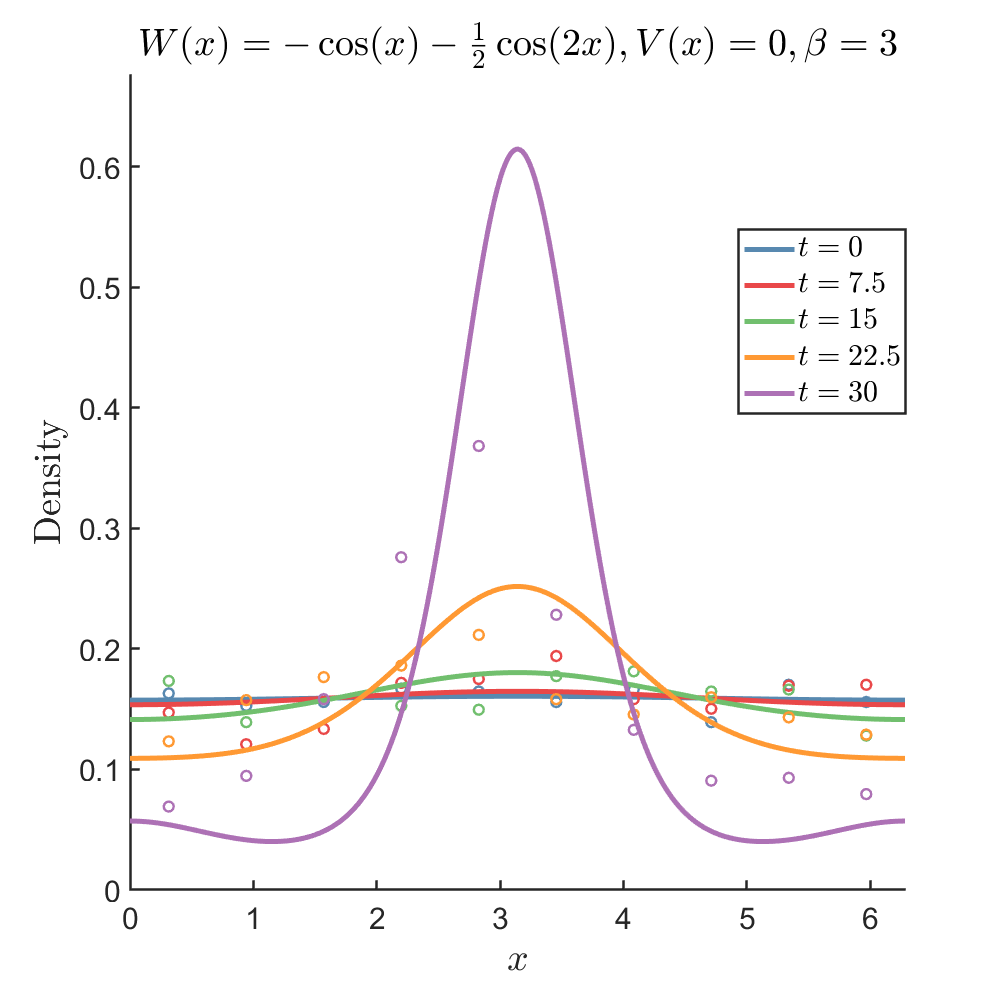}
    \end{minipage}
  \hfill
  \begin{minipage}[b]{0.49\textwidth}
    \includegraphics[width=\textwidth]{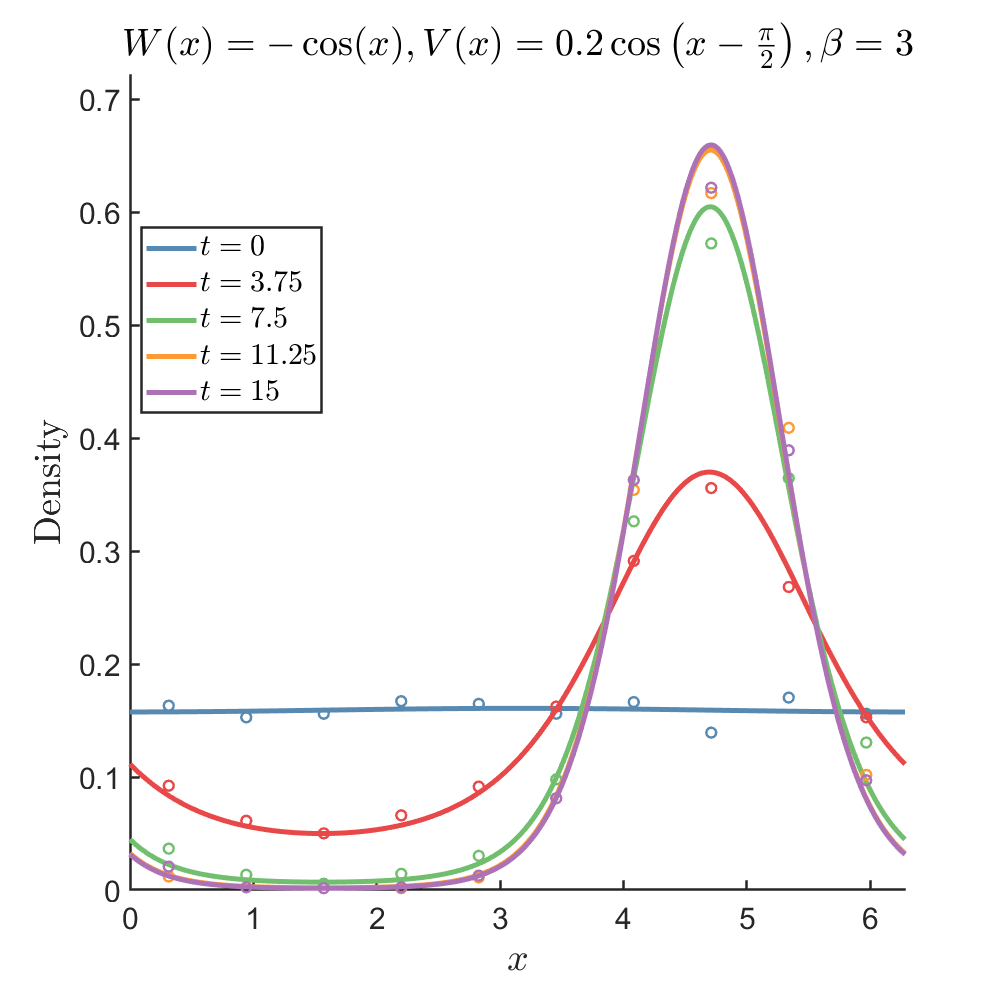}
  \end{minipage}
    \caption{PDE and SDE dynamics for a different range of times, $\beta = 3$, for $W(x) = -\cos(x) - \frac{1}{2}\cos(2x), V(x) = 0$ (left) and for $W(x) = -\cos(x)$, $V(x) = 0.2 \cos \left( x - \frac{\pi}{2} \right)$ (right).}\label{fig:Snapshots}
\end{figure}

\subsection{PDE simulations for the HKB model}
Another model with confining potential is the Haken–Kelso–Bunz (HKB) {which is a well-known system in biophysics} (see Chapter 5, \cite{frank04}). Here $W(x) = -\cos(x)$, $V(x) = - \alpha \cos(x+\pi) - \gamma \cos(2(x+\pi))$. 
{The simulations we present here offer insight into the model which is not completely analytically tractable. Unlike previous models, the HKB model can exhibit either one stable state or two stable states, depending on the system parameters.}
{In Figures~\ref{fig:HKB1} and~\ref{fig:HKB2}} we plot the stationary states of the McKean-Vlasov PDE for different values of $\alpha, \gamma$.
\begin{figure}[H]
    \centering
    \includegraphics[width=0.75\textwidth]{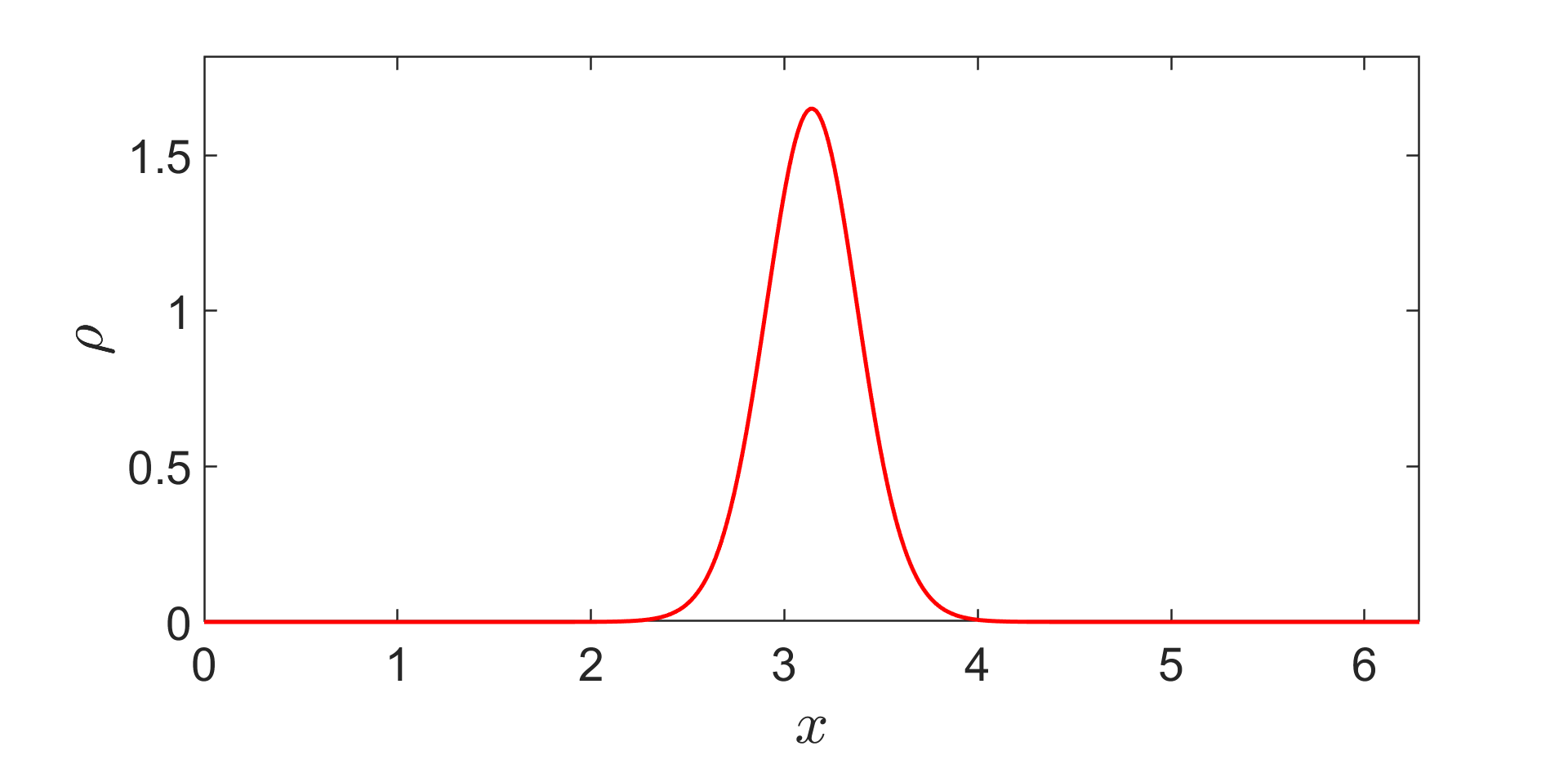}
    \caption{PDE dynamics for the HKB model with $\alpha, \gamma = 1$, $\beta = 3$.}
    \label{fig:HKB1}
\end{figure}

\begin{figure}[H]
    \centering
    \includegraphics[width=0.75\textwidth]{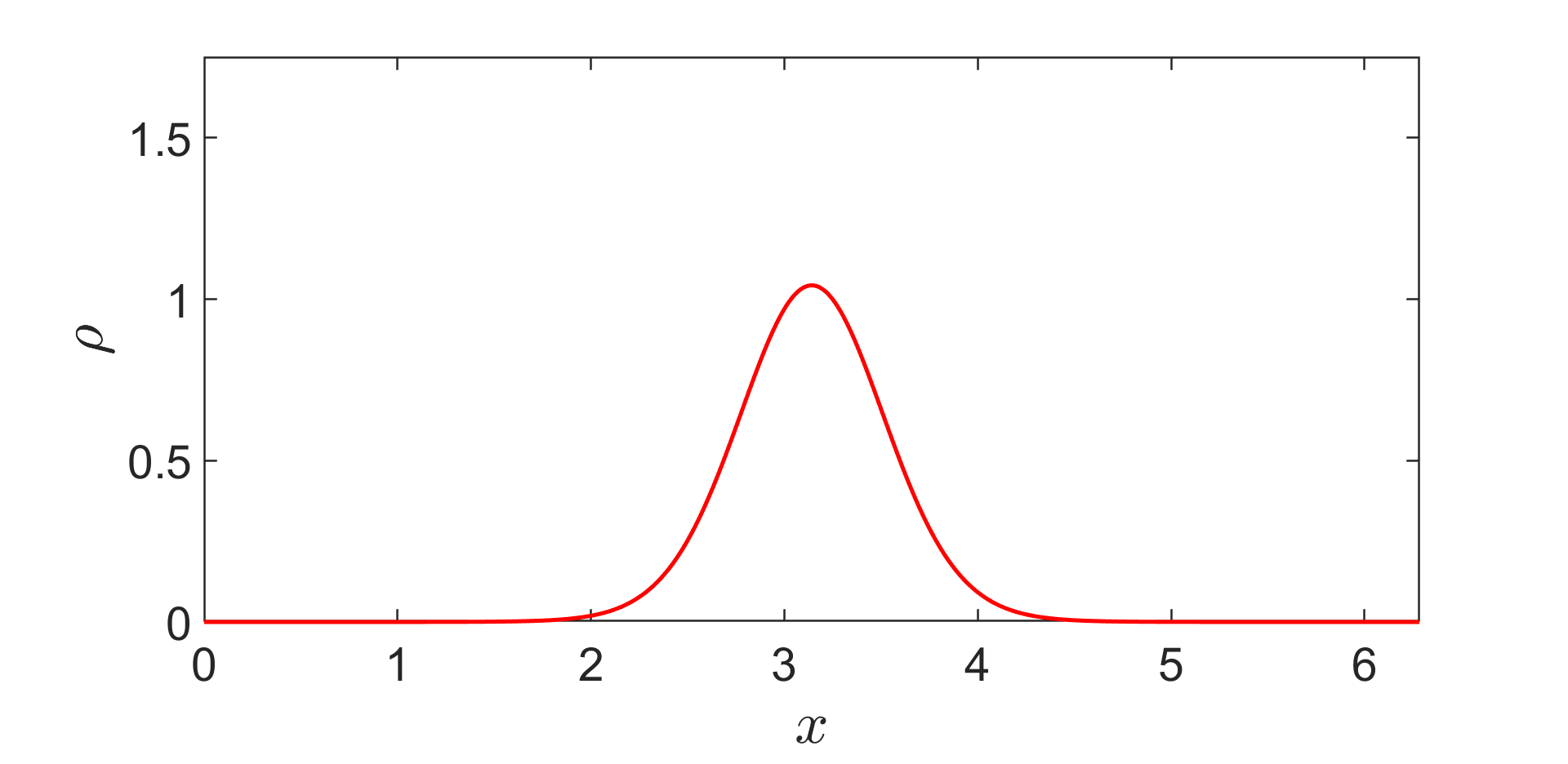}
    \caption{PDE dynamics for the HKB model with $\alpha = 1, \gamma = \frac{1}{8}$, $\beta = 3$.}
    \label{fig:HKB2}
\end{figure}

%\textcolor{blue}{BEN: is it possible to show an example where there are two stationary states?  I think we need to justify this section a bit more.}

%\appendix
%\section{Appendix}
%\subsection{Uniqueness of eigenmodes}
%We show that the problem (\ref{evalue}) does not have any other eigenmodes except for $q_k(x) = A_k \sin(kx)$ or $q_k(x) = B_k \cos(k x)$. Suppose $q$ solves (\ref{evalue}). Expand $q$ as $q(x) = \sum_{l = - \infty}^{\infty} q_l e^{ilx}$.

%For a fixed $k$, we have:
%\begin{align*}
%    &\int_0^{2\pi} q(y) \cos(k(x-y)) \dd y = \sum_{l} q_l \int_0^{2\pi} e^{ily} \cos(k(x-y)) \dd y \\
%    &= \sum_l q_l \int_0^{2\pi} (\cos(ly) + i \sin(l y)) \cos(k(x-y)) \dd y \\
%    &= q_k (\cos(kx) \pi + i \pi \sin(kx)) = \pi q_k e^{ikx}
%\end{align*}

%Therefore, the integral term in (\ref{evalue}) becomes:
%\begin{align*}
%    &\frac{1}{T} \int_0^{2\pi} q(y) \left( \sum_{k=1}^n k^2 a_k \cos(k (x - y)) \right) \dd y = \frac{\pi}{T} \sum_{k=1}^n k^2 a_k q_k e^{ikx}
%\end{align*}

%Therefore $q$ must solve:
%\begin{align*}
%    -2\pi \sum_{k=-\infty}^{\infty} k^2 q_k e^{ikx} + \frac{\pi}{T}  \sum_{k=-\infty}^{\infty} a_k k^2 q_k e^{ikx} = - 2 \lambda \sum_{k=-\infty}^{\infty} q_k e^{ik x}
%\end{align*}
%This needs to be completed.

%{\bf Ben comment:} A related point that we could make is that the eigenvalues are (doubly) repeated and the code will find two orthogonal linear combinations of $\sin(kx)$ and $\cos(kx)$ with the corresponding eigenvalue.

\section{Conclusions}
\label{sec:conclusions}

In this paper, we studied the stability of steady states for McKean-Vlasov dynamics on the torus, for interaction potentials that contain multiple non-zero Fourier modes. It was shown, through the study of both the self-consistency equations and the free energy of the system and through a linear stability analysis, that there is a critical temperature value at which the constant stationary state of the system becomes unstable. Furthermore, the stability of non-constant stationary states was analysed by means of perturbation theory for higher harmonic and multichromatic interaction potentials. Finally, we verified these results with extensive numerical simulations of both the McKean-Vlasov PDEs and the systems of SDEs involved. 

The work presented in this paper can be extended in several interesting directions. First, we would like to study the effect of inertia on the formation and stability of multipeak solutions by considering the kinetic mean field PDE \cite{GPY2018}. Second, the impact of colored noise on the stability on the phase transitions and on the stability of non-uniform steady states is an important question, motivated by recent work on the modeling of collective organization in cyanobacteria~\cite{PhysRevLett.131.158303}. More generally, we aim at applying our analytical and numerical methodologies to the study of active matter, e.g.~\cite{peruani2008mean}. All these topics are currently under investigation.

\paragraph{Acknowledgments} 
We thank the reviewers for a very careful reading of our paper and for many useful suggestions and comments. 

\paragraph{Funding} 
Imperial College London EPSRC DTP in Mathematical Sciences Grant No. EP/W523872/1 and Mary
Lister McCammon summer research fellowship (to B.B.); ERC-EPSRC Frontier Research Guarantee
through Grant No. EP/X038645, ERC through Advanced Grant No. 247031 and 1 and a Leverhulme
Trust Senior Research Fellowship, SRF\textbackslash R1 \textbackslash 24105 (to G.P.).

\paragraph{Conflicts of interest}
The authors declare that they have no conflicts of interest.

\paragraph{Data availability}
No further data is available.

\printbibliography

\end{document}